\documentclass[11pt,twoside,letter,reqno]{amsart}
\usepackage{graphicx} 
\usepackage{fancyhdr}
\usepackage{color}
\usepackage{mathptmx}
\usepackage{mathtools}
\mathtoolsset{showonlyrefs}
\usepackage{algorithm}
\usepackage[noend]{algpseudocode}
\usepackage{url}
\usepackage{subcaption}
\newcommand{\N}{\mathbb{N}}
\newcommand{\Z}{\mathbb{Z}}
\newcommand{\R}{\mathbb{R}}
\newcommand{\C}{\mathbb{C}}
\newcommand{\F}{\mathcal{F}}
\newcommand{\M}{\mathcal{M}}
\newcommand{\intd}{\,\mathrm{d}}
\newcommand{\e}{\mathrm{e}}
\newcommand{\Tcal}{\mathcal{T}}
\newcommand{\Ucal}{\mathcal{U}}
\DeclareMathOperator{\sinc}{\mathrm{sinc}}
\DeclareMathOperator{\supp}{\mathrm{supp}}
\DeclareMathOperator{\grad}{\mathrm{grad}}
\DeclareMathOperator{\diag}{\mathrm{diag}}
\DeclareMathOperator*{\argmin}{\mathrm{argmin}}
\DeclareMathOperator*{\argmax}{\mathrm{argmax}}
\newtheorem{theorem}{Theorem}
\newtheorem{remark}{Remark}
\newtheorem{lemma}{Lemma}

\makeatletter
\let\uppercasenonmath\@gobble%
\let\MakeUppercase\relax%
\let\scshape\relax%
\makeatother

\makeatletter
\def\specialsection{\@startsection{section}{1}%
  \z@{\linespacing\@plus\linespacing}{.5\linespacing}%
  {\normalfont}}%
\def\section{\@startsection{section}{1}%
  \z@{.7\linespacing\@plus\linespacing}{.5\linespacing}%
  {\normalfont\scshape\bfseries}}
\makeatother

\renewcommand{\headrulewidth}{0pt}
\begin{document}

\title{\bf\vspace{-39pt}Super-Resolution for Doubly-Dispersive Channel Estimation}
\author{
  \normalsize Robert Beinert\\
  \small Institute of Mathematics, Technische Universität Berlin\\
  \small Stra\ss{}e des 17.~Juni~136, Berlin, 10623, Germany\\
  \small beinert@math.tu-berlin.de\\
  \\
  \normalsize Peter Jung\\
  \small Communications and Information Theory Group,
  Technische Universität Berlin\\
  \small Einsteinufer~25, Berlin, 10587, Germany\\
  \small peter.jung@tu-berlin.de\\
  \\
  \normalsize Gabriele Steidl\\
  \small Institute of Mathematics, Technische Universität Berlin\\
  \small Stra\ss{}e des 17.~Juni~136, Berlin, 10623, Germany\\
  \small steidl@math.tu-berlin.de\\
  \\
  \normalsize Tom Szollmann\\
  \small Einsteinufer~25, Berlin, 10587, Germany\\
  \small t.szollmann@campus.tu-berlin.de}
\date{27. January 2021}
\maketitle
\thispagestyle{fancy}

\markboth{\footnotesize \rm \hfill R. BEINERT, P. JUNG, G. STEIDL, T. SZOLLMANN \hfill}
{\footnotesize \rm \hfill SUPERRESOLUTION FOR DOUBLY-DISPERSIVE CHANNEL ESTIMATION \hfill}

\begin{abstract} 
  In this work we consider the problem of identification and
  reconstruction of doubly-dispersive channel operators which are
  given by finite linear combinations of time-frequency shifts. Such
  operators arise as time-varying linear systems for example in radar
  and wireless communications. In particular, for information
  transmission in highly non-stationary environments the channel needs
  to be estimated quickly with identification signals of short
  duration and for vehicular application simultaneous high-resolution
  radar is desired as well. We consider the time-continuous setting
  and prove an exact resampling reformulation of the involved channel
  operator when applied to a trigonometric polynomial as identifier in
  terms of sparse linear combinations of real-valued atoms. Motivated
  by recent works of Heckel et al.  we present an exact approach for
  off-the-grid superresolution which allows to perform the
  identification with realizable signals having compact support.  Then
  we show how an alternating descent conditional gradient algorithm
  can be adapted to solve the reformulated problem.  Numerical
  examples demonstrate the performance of this algorithm, in
  particular in comparison with a simple adaptive grid refinement
  strategy and an orthogonal matching pursuit algorithm.
  \vspace{5mm} 
\vspace{3mm}
\end{abstract}

\section{Introduction}\label{sec:intro}
\noindent
Sensing and information retrieval in highly non-stationary
environments are challenging inverse problems in radar and sonar
applications, and their fundamental understanding is also required for
future wireless communication in very rapidly time-varying mobile
scenarios.  In such problems, the task is to identify or estimate
channel parameters in a robust manner by probing the channel with a
particular identifier signal $w$ of finite duration, also called pilot
signal.  In radar, for example, a known radar waveform is transmitted
and from the received reflections, distance and relative velocity of a
target can be obtained by estimating delay and Doppler shifts.
Several reflections superimpose at the receiver, hence the core task
consists in estimating the multiple time-frequency shifts from
finitely many samples of the received signal:
\begin{equation}
  y(t) = \sum_{s=1}^S \eta_s w(t - \tau_s) \e^{2 \pi i \nu_s t}
\end{equation}
taken within a finite observation interval.  Here each triplet
$(\eta_s, \tau_s, \nu_s)$ can be interpreted as a particular
transmission path with a delay $\tau_s$ and Doppler-shift $\nu_s$ due
to relative distance and velocity, respectively, with a complex-valued
attenuation factor $\eta_s$. This so called tapped delay-line model,
is a special case of a doubly-dispersive (or linear time-variant)
channel, where the spreading function is a (finite) point measure.
For more details on this terminology, see for example classical works
\cite{Bello69,Kailath62}.  Intuitively, it is clear that simultaneous
accuracy in time and frequency are governed by the uncertainty
relation and that the shape of the waveform should fit time and
frequency dispersion of the channel. However, often only few
scatterers are affecting the wave propagation and therefore the number
of time-frequency shifts is rather small compared to the number of
samples one may acquire at the receiver.

In so-called coherent communication the wireless channel needs to be
estimated to equalize unknown data signals consecutively or
simultaneously transmitted with the pilot signal. This principle is
used for example in orthogonal frequency-division multiplexing (OFDM)
modulation scheme \cite{Chang66} which is implemented in many of
today's communication technologies like WiFi, LTE and 5G standards, as
well as broadcasting systems like DAB and certain DVB standards
\cite{Kumar15}.  Thus, the first goal here is to estimate the action
of the channel operator on a particular restricted class of data
signals. A channel which is exclusively time- or frequency-selective,
reduces to convolutions or multiplication operators and equalization
(inverting the action of the channel) is then often possible via
conventional deconvolution techniques. In the doubly-selective case
however, more advanced equalization approaches are necessary to deal
with self-interference effects. For this purpose the delay-Doppler
shifts are usually approximated to lie on a-priori fixed lattices
leading to leakage effects \cite{CSPC11}.  In essence, the intrinsic
sparsity of channel does not carry over to the approximated model,
rendering compressed sensing methods like \cite{THER10,Pfander2010}
much less effective.

In radar instead it is important to achieve high resolution on the
time-frequency shift parameters itself.  However, in future high
mobility vehicular communication \cite{KCGH18} and automotive
applications both aspects will become relevant, i.e., discover the
instantaneous neighborhood using radar and simultaneous communicating
with other vehicles or road side units. In particular, combined radar
and communication transceivers which simultaneously shall use the same
hardware and frequency band for both tasks, are recently proposed and
investigated in the literature, see exemplary
\cite{Liu:TCOMM20}. However, since the propagation environment may
change in such vehicular applications as well on a short time-scale
and usually in an almost unpredictable manner, it is also important to
perform channel estimation and radar in short time cycles with short
signals. The traffic type in automotive applications also enforces to
ensure strict latency requirements in communication for decoding the
equalized data signals.

Beside the practical needs for advanced signal processing algorithms
in this challenging engineering field, the estimation problem itself
has been attracted researchers working in harmonic analysis.  First
works in this field and from the perspective of channel identification
are due Pfander et al. \cite{Pfander06:opid}.  Identifiying a linear
operator with restricted spreading, i.e., with bandlimited symbol has
been investigated in \cite{KrahmerPfander14}.

Finally, we like to mention that there exist other methods for
superresolution as Prony-like methods
\cite{CH2015,KMPO2018,LF2016,PT2013,SP2020}.  These are spectral
methods which perform spike localization from low frequency
measurements.  They not need any discretization and recover the
initial signal as long as there are enough observations.  So far we
have not examined if and how such methods could be applied for our
specific modulation-translation setting.
\\

\noindent
{\bfseries Main contribution.}
The main contribution of this paper is twofold. First, we establish an
exact sampling formula for operators which are sparse complex linear
combinations of modulation and translation operators
$$
H = \sum_{s=1}^S \eta_s M_{\nu_s}T_{\tau_s}
$$
applied to (truncated) trigonometric polynomials $w$ as identifiers.
The basic resampling idea goes back to the work of Heckel et
al. \cite{HMS14}, where the problem to identify the parameters
$\eta_s$, $\nu_s$, $\tau_s$ of the unknown operator $H$ is
approximated by a discrete formulation without explicitly accounting
for the employed function spaces and by applying an approximate
sampling formula.  Using trigonometric polynomials as identifiers, we
derive an explicit resampling formula for the continuous problem such
that we can completely avoid the approximation errors in \cite{HMS14}.
By this, we also overcome particular parameter limitations in the
original proof since we not directly couple time-bandwidth limitation
of operator and the identifier.

As a second main result we provide explicit algorithmic reconstruction
approaches.  Our sampling reformulation allows the straightforward
application of standard modifications of the conditional gradient
method, also known as Frank-Wolfe algorithm, to determine the
amplitudes $\eta_s \in \C$ and the two-dimensional positions
$(\tau_s,\nu_s)$.  Here we focus on the alternating direction
conditional gradient (ADCG) algorithm propose by Boyd et
al. \cite{BSR17}.  The corresponding optimization problem takes noise
into account and penalizes the sparsity of the above linear
combination by the $\ell_1$-norm of the amplitudes.  The optimization
problem can be rephrased in terms of atomic measures, where the
$\ell_1$-norm is directly related to the total variation norm of the
measure, resp. to the atomic norm of a certain set of atoms.  Such
problems are known as BLASSO \cite{DGHL2017}.  Besides Frank-Wolfe
like algorithms that minimize the location parameters over a
continuous domain, a common approach consists in constraining the
locations to lie on a grid.  This leads to a finite dimensional convex
optimization problems, known as LASSO \cite{Tib1996} or basis pursuit
\cite{CDS1998} , for which there exist numerous solvers
\cite{CW2005,DDM2004,EHJT2004,TL2008}.  We will compare the ADCG
applied to our resampled problem with a grid method, where we
incorporate an adaptive grid refinement.  As a third group of methods,
we like to mention the reformulation of the optimization problem via
its dual into an equivalent finite dimensional semi-definite program
(SDP).  This technique was first proposed in \cite{CandesGranda12} and
then adapted by many other authors.  However, the equivalence of
formulations are only true in the one-dimensional setting and in
higher dimensions one needs to use e.g. the so-called Lassere
hierarchy \cite{DGHL2017}.  An SDP approach for our two-dimensional
setting based on a results of \cite{Dumitrescu07} was also proposed in
the paper of Heckel et al. \cite{HMS14}.  Since this approach appears
to be highly expansive both in time and memory requirement and has
moreover to fight with many non specific local maxima related to the
so-called dual certificate, it is not appropriate for our setting.

This paper is organized as follows: In Section \ref{sec:prelim}, we
collect the basic notation and results from Fourier analysis and
measure theory which are needed in the following sections. At the end
of the section we establish a theorem which relates trigonometric
polynomials with periodic functions arising from the Fourier transform
of compactly supported measures.  The proof of the theorem is given in
Appendix \ref{app:identifier}.  In Section \ref{sec:main}, we
formulate our superresolution problem for doubly-dispersive channel
estimations.  More precisely, we are interested in the two-dimensional
parameter detection of sparse linear combinations of
translation-modulation operators.  Instead of treating the original
problem, we give a sampling reformulation of the involved
translation-modulation operators for identifiers which are
trigonometric polynomials.  Here the relation between these
polynomials and Fourier transforms of measures will play a role.
Since the identifiers have only to evaluated at points lying in a
compact interval, our choice implies no restriction for practical
purposes.  In Section \ref{sec:main}, we prove the sampling theorem
for translation-modulation operators applied to trigonometric
polynomials. Then, in Section \ref{sec:algs}, we show how an
alternating descent conditional gradient algorithm can be applied to
solve the reformulated problem.  Finally, we demonstrate the
performance of this algorithm in comparison with simple adaptive grid
refinement algorithm and an orthogonal matching pursuit method in
Section \ref{sec:numerics}.

\section{Preliminaries}\label{sec:prelim}
\noindent
{\bfseries Function spaces.}
Let $I$ be an open finite interval of $\mathbb R$ or $\mathbb R$
itself.  By $C(I)$ we denote the space of complex-valued continuous
functions on $I$, by $C_b(I)$ the Banach space of bounded functions
endowed with the norm $\|f\|_\infty = \sup_{x \in I} |f(x)|$.
Further, let $C_0(\R) \subset C_b(\R)$ be the closed subspace of
functions vanishing at infinity.  Let $L^p(I)$, $p \in[1,\infty]$ be
the Banach space of (equivalence classes) of complex-valued Borel
measurable functions with finite norm
\begin{equation} \label{L_p-norm}
  \|f\|_p = \begin{cases}
    \left( \int_{I} |f(x)|^p \intd x \right)^\frac{1}{p}, 
    & 1 \leq p < \infty,\\
    \mathrm{ess}	\sup_{x \in I} |f(x)|, 
    & p = \infty.
  \end{cases}
\end{equation}
For compact $I$, it holds
$L^1(I) \supset L^r(I) \supset L^s(I) \supset L^\infty (I)$, $r < s$.

An \emph{entire} (holomorphic) \emph{function} $f : \C \to \C$ is of
\emph{exponential type} if there exist positive constants $A, B > 0$
such that
\begin{equation}
  | f(z) | \leq A \e^{B |z|}
  \quad \text{for all $z \in \C$.}
\end{equation}
The \emph{exponential type} of $f$ is then defined as the number
\begin{equation}
  \sigma \coloneqq \limsup_{r \to \infty} \frac{\log M(r)}{r},
  \qquad
  M(r) \coloneqq \sup_{|z| = r} |f(z)|.
\end{equation}
The \emph{Bernstein space} $B_\sigma^p$, $p \in [1,\infty]$ consist of
all functions $f$ of exponential type $\sigma$ whose restriction to
$\R$ belongs to $L^p(\R)$.  Endowed with the $L^p$ norm, $B_\sigma^p$
becomes a Banach space, too.  We will need the following sampling
result of Nikol'skii \cite{Nikol'skii1975}.

\begin{theorem}[{Nikol'skii's Inequality \cite[Thm~6.8]{Higgins99}}]
  \label{thm:nikolskiis_inequality}
  Let $p \in [1, \infty]$. Then, for every $f \in B_\sigma^p$ and
  $a > 0$, we have
  \begin{equation}
    \|f\|_p^p
    \leq \sup_{x \in \R} \left\{ a \sum_{k \in \Z} | f(x - a
      k)|^p \right\}
    \leq (1 + a \sigma)^p \|f\|_p^p.
  \end{equation}
\end{theorem}

{\bfseries Fourier transform of functions.}  The \emph{Fourier
  transform} $\F : L^1(\R) \to C_0 (\R) \subset L^\infty(\mathbb R)$
defined by
$$\F f(\xi) 
:= \int_\R f(x) \e^{-2 \pi i \xi x} \intd x = \lim_{R \rightarrow
  \infty} \int_{-R}^R f(x) \e^{-2 \pi i \xi x} \intd x $$ is a bounded
linear operator.  For $1 < p \le 2$, this operator can be extended as
$\F : L^p(\R) \to L^q (\R)$, $\frac1p + \frac1q = 1$ via the limit in
the norm of $L^q(\mathbb R)$ of
$$
\F f (\xi) = \hat f(\xi) = \lim_{R \rightarrow \infty} \int_{-R}^R  f(x) \e^{-2 \pi i \xi x} \intd x.
$$
By Plancherel's equality, the Fourier transform is an isometry on
$L^2(\R)$.  Note that the Fourier transform of a function
$f \in L^p(\R)$ with $p > 2$ can be defined in terms of tempered
distributions. However, the distributional Fourier transform $\hat f$
does in general not correspond to a function.  A special role plays
the \emph{sinus cardinalis} defined as
$$
\sinc(x) := 
\left\{
\begin{array}{ll}
\frac{\sin(\pi x)}{\pi x}& \mathrm{for} \;  x \neq 0 ,\\
1& \mathrm{for} \;  x = 0.
\end{array}
\right.
$$
The sinc function is in $L^2(\R)$ but not in $L^1(\R)$.  Further, we have
$$
\hat \chi_{[-L,L]} (\xi) = 2L \sinc(2L \xi),
$$
where $\chi_{S}$ denotes the \emph{characteristic function} of a set
$S \subseteq \mathbb R$, i.e., $\chi_{S}(x) = 1$ if $x \in S$ and
$\chi_{S}(x) = 0$ if $x \notin S$.  The counterpart of scaled sinc
functions in the periodic setting are the $N$th Dirichlet kernels
given by
\begin{equation}
D_N(x) = \sum_{k = -N}^N \e^{2 \pi i k x} 
= \frac{\sin\left( (2N+1) \pi x\right)}{\sin (\pi x)}, \qquad x \in \R.
\end{equation}

For arbitrary $f\in L^1(\R)$ with $\hat f\in L^1(\R)$, the
\emph{Fourier inversion formula}
$$
f(x) = (\hat f)^\vee (x) := \int_\R \hat f(\xi) \e^{2 \pi i \xi x} \intd \xi
$$
holds true almost everywhere and, moreover, pointwise if the function
$f$ is continuous.  For two functions $f \in L^1(\mathbb R)$ and
$w \in L^p(\mathbb R)$, $p\in [1,\infty]$, the convolution $f*w$ is
defined almost everywhere by
$$
(f*w)(x) = \int_{\R} f(y) w(x-y) \, \intd y
$$
and is contained in $L^p(\R)$. 
For $p \in [1,2]$, the relation between convolution and Fourier transform 
is given by  $\widehat{f*w} = \hat f \, \hat w$.

For $\sigma >0$ and $p \in [1,\infty]$, we denote by
$\mathrm{PW}_\sigma^p$ the \emph{Paley-Wiener class of functions}
$f : \mathbb C \to \mathbb C$ of the form
$$
f(z) = \int_{-\sigma}^\sigma g(\xi) \e^{2 \pi i z \xi} \intd \xi, \quad z \in \mathbb C,
$$
for some $g \in L^p(-\sigma,\sigma)$.  We have the inclusion
$\mathrm{PW}_\sigma^r \subset \mathrm{PW}_\sigma^s$ for $1 \le s < r$.
Functions of the class $\mathrm{PW}_\sigma^p$ are holomorphic and of
exponential type $2 \pi \sigma$ by
\begin{equation}
  \label{eq:holomorphic_extension:2}
  |f(z)|
  \leq 	
  \int_{-\sigma}^\sigma 
  |\hat f(\xi)| 
  \e^{2 \pi |\xi| |z|} 
  \intd \xi
  \leq 
  \|\hat f\|_1
  \e^{2 \pi \sigma |z|},
  \qquad z \in \C.
\end{equation}
For $p \in [1,2]$, we further have
$\mathrm{PW}_{\sigma}^p\subset B^q_{2 \pi \sigma}$, 
see \cite{Higgins99}. 

{\bfseries Measure spaces.}  Let $X$ be a compact subset of $\R^d$ or
$\R^d$ itself.  By $\mathcal M(X)$ we denote all regular, finite,
complex-valued measures, i.e., all mappings
$\mu: \mathcal{B}(X) \rightarrow \mathbb C$ from the Borel
$\sigma$-algebra of $\mathbb R^n$ to $\mathbb C$ with
$|\mu(X)| < \infty$ and
$$\mu \left(\bigcup_{k=1}^\infty B_k \right) = \sum_{k=1}^\infty \mu(B_k)$$
for any sequence $\{B_k\}_{k \in \N} \subset \mathcal{B}(X)$ of
pairwise disjoint sets.  We suppose that the series on the right-hand
side converges absolutely, so that the indices of the sets $B_k$ can
be arbitrarily reordered.  The \emph{support of a complex measure}
$\mu\in \mathcal M(X)$ is defined by
\begin{equation}
  \supp(\mu) = \supp(\rho^+) \cup \supp(\rho^-) \cup \supp(\iota^+)
  \cup \supp(\iota^-),
\end{equation}
where $\rho^+ - \rho^- = \Re(\mu)$ and $\iota^+ - \iota^- = \Im(\mu)$
are the Hahn decompositions of the real and imaginary part into
non-negative measures.  The \emph{support of a non-negative measure} $\nu$
is the closed set
\[\supp(\nu) \coloneqq \bigl\{ x \in X: B \subset X \text{ open, } x \in B  \implies \nu(B) >0\bigr\}.\]
The \emph{total variation of a measure} $\mu \in \mathcal
M(X)$ is defined by
\[
|\mu|(B) \coloneqq \sup \Bigl\{ \sum_{k=1}^\infty |\mu(B_k)|:
\bigcup\limits_{k=1}^\infty B_k = B, \, B_k \; \mbox{pairwise disjoint}\Bigr\}.
\]
With the norm $\| \mu\|_{\mathcal M} \coloneqq |\mu|(X)$ the space
$\mathcal M(X)$ becomes a Banach space.  The space
$\mathcal M(X)$ can be identified via Riesz's representation
theorem with the dual space of $C_0(X)$ and the weak-$\ast$
topology on $\mathcal M(X)$ gives rise to the \emph{weak
  convergence of measures}.  

We will need that, for a bounded
Borel-measurable function $g$, the measure $g \mu$ defined by
$g \mu(B) := \int_B g(x) \intd \mu(x)$ for open $B \subset \R^d$ is again in
$\M(\R^d)$ and
$\| g  \mu \|_{\mathcal M} \leq \|g\|_\infty
\|\mu\|_\mathcal{M}$.

{\bfseries Fourier transform of measures.}
For our purposes, it is enough to consider the Fourier transform of measures on $X = \mathbb R$.
If we consider the open balls
$B_R \coloneqq \{x : |x| < R\}$ of radius $R > 0$, then
\begin{equation}   \label{eq:mono-total-var}
  |\mu|(B_R) \to \| \mu\|_{\mathcal M(\mathbb R)}
  \qquad\text{and}\qquad
  |\mu|(\mathbb R \setminus B_R) \to 0
  \qquad
  \mathrm{as} \quad R\to \infty.
\end{equation}
Indeed, the integral with respect to a measure $\mu\in \mathcal M(\mathbb R)$ is
also well defined for every 
$\varphi \in C_b(\mathbb R)$ and 
$$
  \langle \mu,\varphi\rangle  \le \| \mu \|_{\mathcal M(\mathbb R)} \,
  \|\phi\|_{\infty}.
$$
Consequently, we can define
the \emph{Fourier transform} 
  $\mathcal{F} \colon \mathcal M(\mathbb R) \to C_{b}(\mathbb R)$ by
 $$
    \mathcal{F}\mu( \xi)
    \coloneqq
    \hat{\mu} ( \xi)
    \coloneqq
    \langle \mu, \e^{-2 \pi i x \xi}\rangle 
		=    \int_{\mathbb R} \e^{-2 \pi i  x \xi} d\mu(x).
 $$
The Fourier transform is a linear, bounded
operator from $\mathcal M(\mathbb R)$ into $C_b(\mathbb R)$ with operator norm one.
Moreover, it is unique in the sense that $\mu \in \mathcal M(\mathbb R)$ with 
$\hat \mu \equiv  0$ implies that $\mu$ is the zero measure.
We are especially interested in the Fourier transform of \emph{atomic measures}
$\mu \coloneqq \sum_{k\in \mathbb Z} \alpha_k \delta(\cdot - t_k)$ with
$\alpha_k \in\mathbb C$, $t_k \in \mathbb R$ given by
$$\hat \mu(\xi) = \sum_{k \in \mathbb Z} \alpha_k\, \e^{-2\pi i \xi
  t_k}.$$
If the point masses are equispaced located at $t_k = \frac kn$ with $n
\in \mathbb N$, the
Fourier transform becomes an $n$-periodic Fourier series.  Moreover,
restricting the support of $\mu$ to $[-\sigma, \sigma]$, we obtain the
$n$-periodic trigonometric polynomial
$$\hat \mu(\xi) = \sum_{k=-N}^N \alpha_k\, \e^{-2\pi i \frac{\xi
    k}{n}},$$ where $N = \lfloor \sigma n \rfloor$ and
$t_k = \frac{k}{n}$, $k=-N,\ldots,N$.  The following theorem shows
that also the reverse direction is true, i.e., every periodic function
given as the Fourier transform of a compact measure is a finite
trigonometric polynomials.

\begin{theorem}\label{lem:bandlimit_periodic_stieltjes}
Let $f = \hat \mu_f$ with $\mu_f \in \M(\R)$ 
fulfill $\supp \mu_f \subseteq [-\sigma, \sigma]$ for some $\sigma > 0$.
Suppose that $f$ is $n$-periodic for $n > 0$.
Then $f$ is a trigonometric polynomial of the form
\begin{equation}\label{eq:bandlimit_periodic_stieltjes:1}
f(\xi) 
=
	\sum_{k = -N}^N 
		\hat f(k)
		\e^{2 \pi i \frac{k \xi}{n}},
\qquad \hat f(k) \coloneqq \frac 1n \int_0^n f(\xi) \e^{-2 \pi i \frac{k t}{n}} \intd t,
\end{equation}
where $N = \lfloor \sigma n \rfloor$.
\end{theorem}

The proof of the theorem is given in Appendix \ref{app:identifier}.

\section{Superresolution in Doubly-Dispersive Channel Estimation}
\label{sec:main}
\noindent
In doubly-dispersive channel estimation we are both interested in the
detection of shifts and modulations of signals.  Recall that the
\emph{shift operator} $T_\tau$ and the \emph{modulation operator}
$M_\nu$ are defined for $x, \tau, \nu \in \R$ by
$$T_\tau f(x) := f(x - \tau) \quad \mathrm{and} \quad M_\nu f(x) := f(x) \e^{2 \pi i \nu x},$$
respectively. Their concatenation is given by
$$
M_\nu T_\tau f(x) 
= 
\e^{2\pi i \nu x} f(x-\tau) \quad \text{and} \quad T_\tau M_\nu f(x) = \e^{2\pi i \nu (x-\tau)} f(x-\tau).
$$
Similarly, 
for $f \in L^p(\mathbb R)$ with $p \in [1,2]$, it holds
\begin{align}
\widehat{T_\tau f} = M_{-\tau} \hat f
\quad \text{and} \quad
\widehat{M_\nu f} = T_\nu \hat f.
\end{align}
Both operators are unitary on $L^2(\mathbb R)$.  Note that a similar
definition of shifts and modulations can be given for tempered
distributions, see, e.g., \cite[Section 4.3.1]{PPST2019}.  For
$S \in \N$ and $\mathcal T, \Omega >0$, we consider the operator
\begin{equation} \label{operator}
H := \sum_{s=1}^{S} \eta_s M_{\nu_s} T_{\tau_s}, 
\qquad 
\eta_s \in \C, 
\tau_s \in \left[-\tfrac{\mathcal T}{2},\tfrac{\mathcal T}{2} \right], 
\nu_s \left[-\tfrac{\Omega}{2},\tfrac{\Omega}{2} \right].
\end{equation}
We are interested in the following \textbf{superresolution problem}: 
for a known function $w \in C_b(\mathbb R)$, determine the amplitudes $\eta_s \in \mathbb C$
and positions $\tau_s, \nu_s \in X$, $s=1,\ldots,S$ 
from certain samples of 
\begin{equation} \label{operator_1}
H w = \sum_{s=1}^{S} \eta_s M_{\nu_s} T_{\tau_s} w.
\end{equation}
In this context, the function $w$ is often called \emph{identifier}.

Our solution will be based on an exact sampling formula of $Hw$ 
which contains sparse linear combination of certain real-valued ``atoms''.
The idea to use such a reformulation for later computations
originates from a paper of Heckel et al. \cite{HMS14}.
However, the approach of those authors uses only an approximate sampling formula 
without given error bound
and not an exact one, see Remark  \ref{rem:heckel}.
The main sampling result is given in the following theorem.

\begin{theorem}[Sampling Formula for Translation-Modulation Operators] \label{thm:heckel_formula}
Choose $\Tcal, \Omega > 0$, $N_1, N_2 \in \N$ and set 
$L_1 \coloneqq 2N_1 + 1$, $L_2 \coloneqq 2N_2 + 1$.
Let 
\begin{equation}\label{eq:heckel_formula:2}
w(x) = \sum_{k=-N_1}^{N_1} w_k \e^{2\pi i \frac{\Omega k x}{L_1}}.
\qquad x\in \R
\end{equation}
be an $\frac{L_1}{\Omega}$-periodic trigonometric polynomial.
Then, we have for $\tau, \nu \in \R$ and $x_j = \frac{\Tcal j}{L_2}$, $j = -N_2, \dots, N_2$ that
\begin{align}\label{eq:heckel_formula:1}
M_\nu T_\tau w \left(x_j \right)
&= 	\sum_{u = -N_1}^{N_1} \sum_{v = -N_2}^{N_2}
		\e^{2 \pi i \frac{x_j v}{\mathcal T}}
		w \left( x_j - \frac{u}{\Omega}\right)
		a(\tau,\nu)_{(u,v)},
\end{align}
with	so-called atoms
\begin{equation} \label{atom}	
		a(\tau,\nu)_{(u,v)}  \coloneqq \frac{1}{L_1 L_2} 
		D_{N_1} \left( \frac{u - \Omega \tau}{L_1} \right)
		D_{N_2} \left( \frac{v - \Tcal \nu}{L_2} \right)
\end{equation}
for $u = -N_1,\ldots,N_1$, $v = -N_2, \ldots,N_2$.
\end{theorem}

The proof of Theorem \ref{thm:heckel_formula} is the content of the next section. 
\\

By Theorem \ref{thm:heckel_formula}, we can rewrite the superresolution problem \eqref{operator_1}
with an identifier of the form \eqref{eq:heckel_formula:2} 
for given samples
$y_j = H w(x_j)$, $ x_j = \frac{\Tcal j}{L_2}$, $j=-N_2,\ldots,N_2$ 
as
\begin{equation} \label{operator_2}
y_j  = H w (x_j) = 
	\sum_{u = -N_1}^{N_1} \sum_{v = -N_2}^{N_2} \e^{2 \pi i \frac{x_j v}{\mathcal T}}
		w \left( x_j - \frac{u}{\Omega}\right)
		\sum_{s=1}^{S} \eta_s a(\tau_s,\nu_s)_{(u,v)}.
\end{equation}
By periodicity of the atoms \eqref{atom}, it makes indeed sense  to restrict ourselves
to 
$$
(\tau, \nu) \in  X \coloneqq 
\left[-\tfrac{\mathcal T}{2},\tfrac{\mathcal T}{2} \right] \times 
\left[-\tfrac{\Omega}{2},\tfrac{\Omega}{2} \right],
$$
and to choose $L_i \ge \mathcal T \Omega$, $i=1,2$.  In this case, all
points $x_j - \frac{u}{\Omega}$ at which the periodic identifier $w$
in \eqref{operator_2} must be evaluated, belong to the interval
$I = (-\mathcal T,\mathcal T)$.  In practice, we therefore have not
really to work with a periodic identifier, but can restrict ourselves
to the compactly supported function $\chi_I w$.

Setting
$$
y \coloneqq (y_j)_{j=-N_2}^{N_2} \quad \mathrm{and} \quad
a(\tau,\nu) \coloneqq \left(a(\tau,\nu)_{(u,v)} \right)_{u=-N_1,v=-N_2}^{N_1,N_2},
$$
where we consider $(u,v)$ arranged as a vector, and introducing the
operator
$$G = \left( G_{j,(u,v)} \right)_{j,(u,v)}: \C^{L_1 L_2} \rightarrow
\C^{L_2}$$ with entries
$$
G_{j,(u,v)} \coloneqq \e^{2 \pi i \frac{x_j v}{\mathcal T}}
		w \left( x_j - \frac{u}{\Omega}\right)
$$
we can rewrite the superresulution problem \eqref{operator_2} as 
\begin{equation} \label{operator_3}
y = G \sum_{s=1}^{S} \eta_s a(\tau_s,\nu_s).
\end{equation}
In practical applications, the measurements $y$ are often corrupted by noise 
so that we finally intend to solve
the regularized problem
\begin{equation} \label{problem_2}
\argmin_{\eta \in \C^S, (\tau,\nu) \in X^S} \|G \sum_{s=1}^{S} \eta_s a(\tau_s,\nu_s) - y\|^2 + \lambda \|\eta\|_{\ell_1}, \qquad \lambda > 0,
\end{equation}
where $\eta = (\eta_s)_{s=1}^S$ and $\tau = (\tau_s)_{s=1}^S$, $\nu = (\nu_s)_{s=1}^S$.
Indeed, we may choose $S$ larger than the number of expected translation-modulations 
and hope that the regularization term enforces the sparsest solution.

\begin{remark}\label{rem:cont}
The above problem is closely related to an inverse problem in the space of measures.
To this end, we consider the linear, continuous operator 
$A: \C^{L_1L_2} \rightarrow \mathcal C(X)$ defined by
$A w := \{ (\tau,\nu) \mapsto \langle a(\tau,\nu) , w \rangle \}$ for $w \in \C^{L_1L_2}$.
Its adjoint $A^*: \mathcal M(X) \rightarrow \C^{L_1L_2}$ is given by
\begin{equation}
A^* \mu \coloneqq \int_X a(\tau,\nu) \, d  \mu(\tau,\nu).
\end{equation}
Then, we may consider the inverse problem
\begin{equation}   \label{eq:tv-tik}
  \min_{\mu \in \mathcal M(X)}
  \frac12 \| G A^* \mu - y \|_2^2
  + \lambda \| \mu \|_{\mathcal M(X)}.
\end{equation}
Problems of this kind are also known as BLASSO \cite{DGHL2017,CandesGranda12} and
were studied in several papers, e.g., by
Bredies and Pikkarainen \cite{BP13} and Denoyelle et al. \cite{DDPS20}.
In particular, it was shown that the problem has a solution.
Since $G A^*$ is not injective, the solution is in general not unique. 
Restricted to atomic measures in $\mathcal M(X)$, i.e.
$\mu = \sum_{s=1}^S \eta_s \delta\left( \cdot - (\tau_s,\nu_s \right))$,
problem \eqref{eq:tv-tik} takes the form \eqref{problem_2}.

The superresolution problem may be also seen from the point of view the so-called atomic norm
formulation addressed in a couple of papers \cite{CRPW12,CF2020,DDPS20,DP17,TBR2015}.
Since $\eta_s = |\eta_s| \e^{2\pi i \phi_s}$ is complex-valued, the set of atoms 
must be redefined as $\{ \e^{2\pi i \phi_s} a(\tau,\nu): \phi \in [0,1), (\tau,\nu) \in X\}$ 
to take real linear combinations
of atoms. 
\end{remark}

As already mentioned, superresolution problem \eqref{operator_1}
has been already considered by Heckel et al. \cite{HMS14}. However, these authors 
proposed to use a different identifier, an issue addressed in the next remark.

\begin{remark}[Relation to the work of Heckel et al. \cite{HMS14}] \label{rem:heckel}
The authors of \cite{HMS14} considered the case $N_1=N_2=N$ and $L_1 = L_2 = L := \Tcal \Omega$,
so that the resampling formula \eqref{operator_2} becomes
\begin{align} \label{eq:derive_heckels_formula:2}	
H w \left(\frac{j}{\Omega} \right)
&= \frac{1}{L^2} 
	\sum_{s=1}^S \eta_s
	\sum_{u = -N}^{N} \sum_{v = -N}^{N}
		\e^{2 \pi i \frac{jv}{L}}
		w \left( \frac{j - u}{\Omega}\right) \\
		& \quad\times
		D_{N} \left( \frac{u - \Omega \tau_s}{L} \right)
		D_{N} \left( \frac{v - \Tcal \nu_s}{L} \right).
\end{align}
However, as identifier they propose
$$
w(x) = \sum_{r=-R}^R \sum_{m=-N}^N w_m \sinc\left( \Omega \left(x- \frac{rL+m}{\Omega} \right)    \right)
$$
with some $R \in \N$. Actually, $R=1$ was applied in \cite{HMS14}.
Since the sinc function is not periodic, the resampling formula
\eqref{eq:derive_heckels_formula:2} does not hold exactly and only
gives an approximation.
\end{remark}

\section{Resampling Results for Translation-Modulation Operators}\label{sec:resampling}
In this section, we prove Theorem \ref{thm:heckel_formula}.  The basis
is the Sampling Theorem \ref{lem:sampling_theorem_l1} for $L^1$
functions.  Then we prove certain sampling formulas which are of
interest on their own.  First, in Lemma
\ref{thm:operator_resampling_l1}, we show a sampling formula for
$p\,H (\hat q * w)$, where $p,q$ are compactly supported functions
with Fourier transform in $L^1(\R)$, for general $w \in L^\infty(\R)$
using certain compactly supported helper functions $\phi$ and $\psi$.
Restricting to identifiers $w$ which are Fourier transforms of
measures, we will see in Theorem
\ref{thm:operator_resampling_stieltjes} that the helper functions can
be avoided.  Finally, we will use this theorem together with
approximation arguments involving sequences of compactly supported
Schwartz functions $\{p_n\}_n$ and $\{q_n\}_n$ to prove Theorem
\ref{thm:heckel_formula}.  We start by recalling a sampling theorem
for $L^1$ functions.

\begin{theorem}[Sampling Theorem for $L^1$-functions] \label{lem:sampling_theorem_l1}
Let $f \in L^1(\R) \cap C_0(\R)$ be a band-limited function with $\supp \hat f \subseteq [-\frac{\Omega}{2}, \frac{\Omega}{2}]$.
Choose $0 < a < 1/\Omega$.
Then for any low-pass kernel 
$\phi \in  L^1(\R) \cap C_0(\R)$ satisfying
\begin{equation}
\hat \phi(\xi) = \begin{cases}
	a, & |\xi| \leq \frac{\Omega}{2},\\
	0, & |\xi| \geq \frac{1}{2a},
\end{cases}
\end{equation}
we have 
\begin{equation}
\label{eq:sampling_theorem_l1:1}
f(x) = \sum_{k \in \Z} f(a k) \phi(x - ak)
\end{equation}
for all $x \in \R$ with absolute and uniform convergence on $\R$ and convergence in $L^1(\mathbb R)$.
\end{theorem}

For convenience, the proof is given in Appendix \ref{app:b}.
In the following, we will further need the next auxiliary lemma.

\begin{lemma} \label{lem:operator_absorption_l1}
Let $w \in L^\infty(\R)$ and $p, q \in L^1(\R)$ with $\hat p, \hat q \in L^1(\R)$.
For $F \in L^1(\R^2)$, we define the linear operator $\mathcal{D}_w : L^1(\R^2) \to L^\infty(\R)$ by 
\begin{equation} \label{eq:operator_absorption_l1:1}
	(\mathcal{D}_w F)(x) 
	:= \iint_{\R^2}
		F(t, \xi) w(x - t) \e^{2 \pi i \xi t}
		\intd t \intd \xi.
\end{equation}
Then $\mathcal{D}_w$ is  continuous and for all $\tau, \nu \in \R$ we have
\begin{equation}
\label{eq:operator_absorption_l1:2}
	p(x) M_\nu T_\tau(\hat q \ast w)(x)
	= \mathcal{D}_w  (T_\tau \hat q \otimes T_\nu \hat p)(x) \qquad \text{for a.e.\ $x \in \R$}.
\end{equation}
\end{lemma}

\begin{proof}
For any $x \in \R$, we have
\begin{equation}
	|(\mathcal D_w F)(x)| 
	\leq \iint_{\R^2} |F(t, \xi)| |w(x - t)| \intd t \intd \xi
	\leq \|w\|_{L^\infty(\R)} \|F\|_{L^1(\R^2)}.
\end{equation}
Thus $\|\mathcal D_w\|_{L^1(\R^2) \to L^\infty(\R)} \leq \|w\|_{L^\infty(\R)}$ and the first claim follows.

For the left-hand side of \eqref{eq:operator_absorption_l1:2} we have
by \emph{Young's convolution inequality}, see \cite{PPST2019}, that
\begin{equation}
  \|\hat q \ast w\|_\infty \leq \|\hat q\|_1 \|w\|_\infty.
  \vspace{5pt}
\end{equation}
Since $\hat p \in L^1(\R)$, we know that $p \in L^\infty(\mathbb R)$.
This implies $p \, M_\nu T_\tau (\hat q \ast w) \in L^\infty(\R)$.
Using that $p(x) = \int_\R \hat p(\xi) \e^{2 \pi i \xi x} \intd \xi$
a.e., we obtain for a.e. $x \in \mathbb R$ that
\begin{align}
p(x) M_\nu T_\tau (\hat q \ast w) (x)
&= p(x) \e^{2 \pi i x \nu} \int_{\mathbb R} \hat q(t) w(x-\tau - t) \intd t\\
&= \int_\R \hat p(\xi) \e^{2 \pi i \xi x} \intd \xi \, \e^{2 \pi i x \nu} \, \int_{\mathbb R} \hat q(t) w(x-\tau - t) \intd t\\
&= \int_\R \int_\R \hat p(\xi) \e^{2 \pi i x (\xi + \nu)} \hat q(t) w(x-\tau - t) \intd t \intd \xi\\
&= \int_\R \int_\R
	\hat q(t -  \tau) \hat p(\xi - \nu) w(x - t) \e^{2 \pi i \xi x}
\intd t \intd \xi \\
&= \mathcal D_w (T_\tau \hat q \otimes T_\nu \hat p) (x).
\label{eq:operator_absorption_l1:4}
\end{align}
\end{proof}

We use the above lemma to show the following intermediate sampling formula.

\begin{lemma} \label{thm:operator_resampling_l1} Let $H$ be given by
  \eqref{operator}.  Let $w \in L^\infty(\R)$ and
  $p, q \in L^1(\R) \cap C_0(\R)$ with $\hat p, \hat q \in L^1(\R)$
  and $\supp p \subseteq [-\frac{\Tcal_p}{2}, \frac{\Tcal_p}{2}]$ as
  well as
  $\supp q \subseteq [-\frac{\Omega_q}{2}, \frac{\Omega_q}{2}]$.
  Choose step-sizes $0 < a < 1 / \Omega_q$ and $0 < b < 1 / \Tcal_p$.
  Then for any $\phi, \psi \in L^1(\R) \cap C_0(\R)$ with
  $\hat \phi, \hat \psi \in L^1(\R)$ obeying
  \begin{equation}
    \psi(x) = \begin{cases}
      b, & \mathrm{for} \; |x| \leq \tfrac{\Tcal_p}{2},\\
      0, & \mathrm{for} \; |x| \geq \tfrac{1}{2b},
    \end{cases}
    \qquad
    \phi(x) = \begin{cases}
      a, & \mathrm{for} \; |x| \leq \tfrac{\Omega_q}{2},\\
      0, & \mathrm{for} \; |x| \geq \tfrac{1}{2a}.
    \end{cases}
  \end{equation}
  we have
  \begin{equation} 		\label{eq:operator_resampling_l1:1}
    p(x) H (\hat q \ast w)(x)
    = 	\psi (x)
    \sum_{k \in \Z} \sum_{\ell \in \Z}
    c_{k, \ell}
    M_{b \ell} T_{a k} (\hat \phi \ast w)(x)
  \end{equation}
  for all $x \in \R$, where 
  \begin{equation}
    c_{k, \ell} 
    := 	\sum_{s=1}^S 
    \eta_s 
    \hat q(a k - \tau_s) 
    \hat p(b \ell - \nu_s),
    \qquad k, \ell \in \Z.
  \end{equation}
  The series on the right side of
  \eqref{eq:operator_resampling_l1:1} converges uniformly on
  $\R$.
\end{lemma}

\begin{proof}
  By linearity it suffices to consider the case $H = M_\nu T_\tau$.
  Since $p,q \in L^1(\mathbb R)$, we have
  $\hat p,\hat q \in C_0(\mathbb R)$ so that
  $F \coloneqq T_\tau \hat q \otimes T_\nu \hat p \in L^1(\R^2) \cap
  C_0(\R^2)$.  Moreover, by the support properties of $p$ and $q$, we
  get
  $\supp \hat F \subset [-\tfrac{\Omega_q}{2}, \tfrac{\Omega_q}{2}]
  \times [-\tfrac{\Tcal_p}{2}, \tfrac{\Tcal_p}{2}]$.  Consequently, we
  can apply Theorem~\ref{lem:sampling_theorem_l1} to $F$ along each
  dimension w.r.t.\ the step-sizes $a$ and $b$ and low-pass kernels
  $\hat \phi$ and $ \hat\psi$ to obtain
\begin{equation}\label{eq:operator_resampling_l1:4}
  F(x, y) 
  = 	\sum_{k\in \Z} \sum_{\ell \in \Z} 
  F(ak, b \ell) 
  \hat \phi(x - ak) 
  \hat \psi(y - b \ell),
\end{equation}
which converges absolutely and uniformly.  
For the $L^1$-convergence,
we have to show that
\begin{align}
  &\int_{\mathbb R} \int_{\mathbb R}
    \biggl|
    \sum_{|k|\ge N} \sum_{|\ell| \ge N} 
    F(ak, b \ell) 
    \hat \phi(x - ak) 
    \hat \psi(y - b \ell)
    \biggr|
    \intd y \intd x
  \\
  &\qquad
    \le
    \int_{\mathbb R} 
    \biggl|
    \sum_{|k|\ge N}  
    \hat q(ak - \tau) 
    \hat \phi(x - ak) 
    \biggr|
    \intd x
    \cdot
    \int_{\mathbb R}
    \biggl|
    \sum_{|\ell| \ge N} 
    \hat p( b \ell - \nu)
    \hat \psi(y - b \ell)
    \biggr|
    \intd y
\end{align}
vanishes for $N \to \infty$, which follows for both integrals as
discussed in the proof of Theorem~\ref{lem:sampling_theorem_l1}.

As the operator $\mathcal D_w : L^1(\R^2) \to L^\infty(\R)$ 
defined in Lemma \ref{lem:operator_absorption_l1} is continuous we conclude
\begin{equation}\label{eq:operator_resampling_l1:5}
p \, H(\hat q \ast w)
= D_w F
= \sum_{k\in \Z} \sum_{\ell \in \Z}
	F(ak, b \ell) 
	D_w (
		T_{ak} \hat \phi 
		\otimes T_{b\ell} \hat \psi
	) \qquad \text{a.e.} 
\end{equation}
By applying Lemma \ref{lem:operator_absorption_l1} once again, we obtain
\begin{equation}\label{eq:operator_resampling_l1:7}
	\mathcal D_w \bigl(
		T_{ak} \hat \phi \otimes T_{b\ell} \hat \psi
	\bigr)(x)
	= \psi (x) \, M_{b\ell} T_{ak} (\hat \phi \ast w)(x)
	\qquad \text{for a.e. $x \in \mathbb R$}.
\end{equation}
Consequently we get for a.e. $x \in \mathbb R$ that
\begin{align}\label{eq:operator_resampling_l1:8}
p(x)\,  H (\hat q \ast w)(x)
&= \sum_{k\in \Z} \sum_{\ell \in \Z}
	F(ak, b \ell) 
	\psi (x)\, 
	M_{b\ell} T_{ak} (\hat \phi \ast w)(x)\\
&= \psi (x) \, \sum_{k\in \Z} \sum_{\ell \in \Z}
	\bigl[ \hat q(ak - \tau) \hat p(b \ell - \nu) \bigr] 
	M_{b\ell} T_{ak} (\hat \psi \ast w)(x).
\end{align}
Note that by Theorem \ref{thm:nikolskiis_inequality} the sequences
$\bigl(\hat q(ak - \tau)\bigr)_{k \in \Z}$ and
$\bigl(\hat p(b \ell - \nu)\bigr)_{\ell \in \Z}$ are absolutely
summable.  The functions $M_{b\ell}T_{ak} (\hat \psi \ast w)$ are
bounded by
\begin{equation}
\|M_{b\ell}T_{ak} (\hat \phi \ast w)\|_\infty
\leq \|\hat \phi \ast w \|_\infty
\leq \|\hat \phi\|_1 \|w\|_\infty.
\end{equation}
Thus, the series \eqref{eq:operator_resampling_l1:8} converges
uniformly on $\R$ and, since the partial sums in
\eqref{eq:operator_resampling_l1:8} are continuous functions, we
conclude that the series converges to a continuous bounded function.
As $p$ and $\hat q \ast w$ are also continuous and bounded, we see
that \eqref{eq:operator_resampling_l1:1} holds for all $x \in \R$.
\end{proof}

Although Theorem~\ref{thm:operator_resampling_l1} works on arbitrary
bounded identifiers $w \in L^\infty(\R)$, the fact that the left side
of \eqref{eq:operator_resampling_l1:1} does not depend on $\phi$ and
$\psi$ suggests that there might be a way to avoid the use of these
functions.  For this purpose, we restrict our attention to a subset of
$L^\infty(\R)$, namely functions $f = \hat \mu_f$ with
$\mu_f \in \mathcal M(\mathbb R)$.  Having the Fourier convolution
theorem in mind, for a Borel measurable, bounded function $\phi$, we
define the convolution
\begin{equation}\label{filter}
(\phi \star_{\mathcal F} f) (x) \coloneqq  \widehat{(\phi \mu_f)} (x) = \int_{\mathbb R} \phi(\xi) \e^{-2\pi i x \xi}  \intd \mu_f(\xi),
\end{equation}
which yields a continuous and bounded function.
If $\phi \in L^1(\mathbb R) \cap C_0(\mathbb R)$ such that $\hat \phi
\in L^1(\mathbb R)$, then our convolution may be expressed by the
Fourier convolution as
$$\phi \star_{\mathcal F} f =  \hat \phi * f.$$
We have the following convergence result.

\begin{lemma} \label{lem:freq_filter_convergence}
Let $f = \hat \mu_f$ with $\mu_f \in \mathcal{M}(\mathbb R)$ and let $g$ be a bounded Borel-measurable function.
Assume that the uniformly bounded and Borel measurable functions $g_n : \R \to \C$
converge pointwise to $g : \R \to \C$.
Then $g_n \star_{\mathcal F} f$ converges uniformly to $g \star_{\mathcal F} f$, i.e.,
\begin{equation}
\|g_n \star_{\mathcal F} f - g \star_{\mathcal F} f \|_\infty \to 0, \qquad \text{as } n \to \infty.
\label{eq:freq_filter_convergence:1}
\end{equation}
\end{lemma}

\begin{proof}
Applying Fatou's lemma, we obtain
\begin{align}
\limsup_{n \to \infty} 
	\|g_n \star_{\mathcal F} f - g \star_{\mathcal F} f \|_\infty
&= \limsup_{n \to \infty} \left\{ \sup_{x \in \R} 
	\bigr| \int_\R 
		(g_n(\xi) - g(\xi)) 
		\e^{-2\pi i x \xi} 
	\intd \mu_f(\xi) \bigr|
\right\}\\
&\leq \limsup_{n \to \infty} \left\{
	\int_\R
		|g_n(\xi) - g(\xi)|
	\intd |\mu_f|(\xi)
\right\}\\
&\leq \int_\R
	\underbrace{\limsup_{n \to \infty} |g_n(\xi) - g(\xi)|}_{= 0}
\intd |\mu_f|(\xi)
= 0.
\end{align}
The lemma of Fatou is applicable since $\|g - g_n\|_\infty \leq 2M$ for some $M > 0$ 
and constant functions are integrable w.r.t. $\mu_f \in \M(\R)$.
\end{proof}

\begin{theorem}\label{thm:operator_resampling_stieltjes}
Let $H$ be given by \eqref{operator}.
Let $w = \hat \mu_w$ with $\mu_w \in \M(\R)$ and 
$p, q \in L^1(\R) \cap C_0(\R)$ with $\hat p, \hat q \in L^1(\R)$ and 
$\supp p \subseteq [-\frac{\Tcal_p}{2}, \frac{\Tcal_p}{2}]$ and $\supp q \subseteq [-\frac{\Omega_q}{2}, \frac{\Omega_q}{2}]$.
Choose $0 < a < 1 /\Omega_q$ and $0 < b < 1 / \Tcal_p$. Then, for all $x \in \R$, we have
\begin{equation} \label{eq:operator_resampling_stieltjes:1}
p(x) H(\hat q \ast w)(x)
= 	ab\, \chi_{(-\frac{1}{2b}, \frac{1}{2b})}(x)\,
	\sum_{k \in \Z} \sum_{\ell \in \Z}
		c_{k, \ell} M_{b \ell} T_{a k} \bigl(\chi_{(-\frac{1}{2a}, \frac{1}{2a})} \star_{\mathcal F} w \bigr)(x),
\end{equation}
where
\begin{equation}
c_{k, \ell}
=	\sum_{s=1}^S 
		\eta_s
		\hat q(a k - \tau_s) 
		\hat p(b \ell - \nu_s), \qquad k, \ell \in \Z.
\end{equation}
The series on the right-hand side of \eqref{eq:operator_resampling_stieltjes:1} converges uniformly on $\R$.
\end{theorem}

\begin{proof}
Let $(\psi_n)_{n \in \Z}$ and $(\phi_n)_{n\in \N}$ be uniformly bounded sequences of Schwartz functions with 
\begin{equation}
	\psi_n(x) = \begin{cases}
		b, & \text{for $|x| \leq \tfrac{\Tcal_p}{2}$},\\
		0, & \text{for $|x| \geq \tfrac{1}{2b}$},
	\end{cases}
	\qquad
	\phi_n(x) = \begin{cases}
		a, & \text{for $|x| \leq \tfrac{\Omega_q}{2}$},\\
		0, & \text{for $|x| \geq \tfrac{1}{2a}$}
	\end{cases}
\end{equation}
for all $n \in \N$ which converge for $n\to \infty$ pointwise as
\begin{equation}
	\psi_n(x) \to \begin{cases}
		b, & \text{for $|x| < \tfrac{1}{2b}$},\\
		0, & \text{for $|x| \geq \tfrac{1}{2b}$},
	\end{cases}
	\qquad
	\phi_n(x) \to \begin{cases}
		a, & \text{for $|x| < \tfrac{1}{2a}$},\\
		0, & \text{for $|x| \geq \tfrac{1}{2a}$}.
	\end{cases}
	\end{equation}
Abbreviating $y \coloneqq p H (\hat q \ast w)$, we obtain by Theorem~\ref{thm:operator_resampling_l1} that
\begin{equation}
\label{eq:operator_resampling_stieltjes:3}
y(x)
= 	\psi_n (x)
	\sum_{k \in \Z} \sum_{\ell \in \Z}
	c_{k, \ell}
	M_{b \ell} T_{a k} (\hat \phi_m \ast w)(x)
\qquad \text{for all $x\in \R$ and $n,m \in \N$.}
\end{equation}
Note that neither $y(x)$ nor $c_{k,\ell}$ depend on $n$ or $m$.
Letting $n \to \infty$, we immediately obtain the pointwise limit
\begin{equation}
\label{eq:operator_resampling_stieltjes:4}
y(x)
= 	b\, \chi_{(-\frac{1}{2b}, \frac{1}{2b})} (x)
	\sum_{k \in \Z} \sum_{\ell \in \Z}
	c_{k, \ell}
	M_{b \ell} T_{a k} (\hat \phi_m \ast w)(x)
\qquad \text{for all $x\in \R$ and $m \in \N$.}
\end{equation}
Now consider the series:
We already used in the proof of Theorem~\ref{thm:operator_resampling_l1} 
that by Theorem \ref{thm:nikolskiis_inequality} 
the coefficients $(c_{k,\ell})_{k,\ell \in \Z} \in \ell^1(\Z^2)$ are absolutely summable.
Moreover, writing $\phi \coloneqq a \chi_{(-\frac{1}{2a}, \frac{1}{2a})}$ we
know by construction that $\phi_m(x) \to \phi(x)$ as $m \to \infty$
for every $x \in \R$
and $(\phi_m)_{m \in \Z}$ is uniformly bounded.
We can therefore apply Lemma~\ref{lem:freq_filter_convergence} to obtain
\begin{equation}
\| \phi_m \star_{\mathcal F} w - \phi\star_{\mathcal F} w\|_\infty \to 0, \qquad \text{as } m \to \infty.
\end{equation}
Since we have $\hat \phi_m \ast w = \phi_m \star_{\mathcal F} w$ for all $m \in \N$ we estimate
\begin{align}
\bigr\| \sum_{k \in \Z} \sum_{\ell \in \Z} 
	c_{k, \ell} 
	M_{b\ell} T_{ak} ( \hat \phi_m \ast w - \phi \star_{\mathcal F} w )
\bigr\|_\infty
&\leq \sum_{k \in \Z} \sum_{\ell \in \Z}
	|c_{k,\ell}|
	\|\hat \phi_m \ast w - \phi \star_{\mathcal F} w \|_\infty\\
&= \|(c_{k,\ell})_{k,\ell\in \Z} \|_1
	\| \phi_m \star_{\mathcal F} w - \phi \star_{\mathcal F} w \|_\infty.
\end{align}
Letting $m \to \infty$ the right side converges to $0$ which proves that
\begin{equation} \label{eq:operator_resampling_stieltjes:5}
y(x)
= 	b\, \chi_{(-\frac{1}{2b}, \frac{1}{2b})} (x)
	\sum_{k \in \Z} \sum_{\ell \in \Z}
	c_{k, \ell}
	M_{b \ell} T_{a k} (\phi \star_{\mathcal F} w)(x)
\end{equation}
for all $x\in \R$, which is equivalent to \eqref{eq:operator_resampling_stieltjes:1}.

The uniform convergence of the series follows immediately from $(c_{k,\ell})_{k,\ell \in \Z} \in \ell^1(\Z^2)$ 
and $\chi_{(-\frac{1}{2a}, \frac{1}{2a})} \star_{\mathcal F} w \in C_b(\R)$.
\end{proof}

Now we can prove our main theorem.

\begin{proof}[Proof of {Theorem~\ref{thm:heckel_formula}}]
  1. Since $\frac{|k|\Omega}{L_1} \leq \frac{(L_1-1)\Omega}{2L_1}$ in
  the reprsentation \eqref{eq:heckel_formula:2} of the identifier $w$,
  we see that
  $\supp \mu_w \subset [-\frac{L_1-1}{2L_1}\Omega,
  \frac{L_1-1}{2L_1}\Omega]$.  Choose
  $\max \{\frac{L_1-1}{L_1}, \frac{L_2-1}{L_2} \} < \beta < 1$ and let
  $(\gamma_n)_{n\in \N}$ and $(\lambda_n)_{n \in \N}$ be sequences of
  positive numbers such that $1 < \gamma_n$ and
  $\beta < \lambda_n < 1$ and $\gamma_n \lambda_n < 1$ for all
  $n \in \N$ which converge to $1$ as $n \to \infty$.  Then, for
  $n \in \N$, define
\begin{equation}
\Tcal_n \coloneqq \gamma_n \Tcal, \qquad
\Omega_n \coloneqq \gamma_n \Omega, \qquad
a_n \coloneqq \frac{\lambda_n}{\Omega}, \qquad
b_n \coloneqq \frac{\lambda_n}{\Tcal},
\end{equation}
as well as the functions
\begin{equation}
w_n(x) 
:= w \left(\frac{x}{\lambda_n}\right) 
\qquad  x \in \R.
\end{equation}
Clearly, we have for all $n \in \N$ that 
$w_n= \hat \mu_{w_n}$, 
where $\mu_{w_n} \in \M(\mathbb R)$ fulfills 
$$
\supp \mu_{w_n} \subset \left[- \frac{(L_1-1)\Omega}{2L_1 \lambda_n}, \frac{(L_1-1)\Omega}{2L_1 \lambda_n} \right] 
\subset 
\left(-\frac{\Omega}{2}, \frac{\Omega}{2}\right).
$$
Further, the function $w_n$ is $a_n L_1$-periodic.
Let $(p_n)_{n \in \N}, (q_n)_{n \in \N}$ be sequences of Schwartz functions with 
\begin{equation}
	p_n(x) = \begin{cases}
		1, & \text{for $|x| \leq \tfrac{\Tcal}{2}$},\\
		0, & \text{for $|x| \geq \tfrac{\Tcal_n}{2}$},
	\end{cases}
	\qquad
	q_n(x) = \begin{cases}
		1, & \text{for $|x| \leq \tfrac{\Omega}{2}$},\\
		0, & \text{for $|x| \geq \tfrac{\Omega_n}{2}$}.
	\end{cases}
\end{equation}
We consider the signal
\begin{equation}
y_n(x) := p_n(x) M_\nu T_\tau (\hat q_n \ast w_n) (x), \qquad x\in \R.
\label{eq:heckel_formula:3}
\end{equation}
Now $p_n, q_n$ as well as
$a_n = \frac{ \lambda_n}{\Omega} < \frac{1}{\Omega_n}$ and
$b_n = \frac{ \lambda_n}{ \Tcal} < \frac{1}{\Tcal_n}$ satisfy the
assumptions of Theorem \ref{thm:operator_resampling_stieltjes}. Hence
we get
\begin{equation}
y_n(x)
= 	a_n b_n 
	\chi_{(-\frac{1}{2b_n}, \frac{1}{2b_n})}(x)
	\sum_{k \in \Z} \sum_{\ell \in \Z}
		c_{n, k, \ell}
		M_{b_n \ell} T_{a_n k} 
			\left( \chi_{(-\frac{1}{2a_n}, \frac{1}{2a_n})} \star_{\mathcal F} w_n \right)(x).
\label{eq:heckel_formula:4}
\end{equation}
with $c_{n, k, \ell} \coloneqq \hat q_n(a_n k - \tau) \hat p_n(b_n \ell - \nu)$ for $k, \ell \in \Z$.

Since $\frac{1}{a_n} = \frac{\Omega}{\lambda_n} > \Omega$ it follows that 
$\supp \mu_{w_n} \subset (-\frac{\Omega}{2}, \frac{\Omega}{2}) \subset (-\frac{1}{2a_n}, \frac{1}{2a_n})$.
Therefore we have for all $x \in \R$ and $n \in \N$ that
\begin{equation}
\chi_{(-\frac{1}{2a_n}, \frac{1}{2a_n})} \star_{\mathcal F} w_n(x) 
= \int_{(-\frac{1}{2a_n}, \frac{1}{2a_n})} \e^{-2\pi i\xi x} \intd \mu_{w_n}(\xi)
= w_n (x).
\end{equation}
Thus for $|x| < \frac{1}{2b_n}$ we can simplify \eqref{eq:heckel_formula:4} to
\begin{equation}
y_n(x)
= a_n b_n \sum_{k \in \Z} \sum_{\ell \in \Z}
		c_{n, k, \ell}
		M_{b_n \ell} T_{a_n k} w_n (x).
\end{equation}
2.
For $j = -N_2, \dots, N_2$, we consider  
\begin{equation}
y_n\left(\frac{j}{b_n L_2}\right) 
= 	a_n b_n 
	\sum_{k \in \Z} \sum_{\ell \in \Z}
		c_{n, k, \ell}
		w_n\left(\frac{j}{b_n L_2} - a_n k\right)
		\e^{2 \pi i \frac{\ell j}{L_2}}.
\label{eq:heckel_formula:5}
\end{equation}
Since $\hat p_n, \hat q_n$ are Schwartz functions, 
we know that $(c_{n,k,\ell})_{k,\ell\in\Z} \in \ell^1(\Z^2)$.
Further $w_n$ is bounded, so that
the series in \eqref{eq:heckel_formula:5} converges absolutely.
Consequently we can rearrange the summation and use the substitution
$k = rL_1 + u$ and $\ell = t L_2 + v$ for $r, t \in \Z$ and 
$u = -N_1, \dots, N_1$ as well as $v = -N_2, \dots, N_2$ to obtain
\begin{align}
  y_n \Bigl(\frac{j}{b_n L_2}\Bigr) 
  &= 	a_n b_n 
    \sum_{r \in \Z} \sum_{t \in \Z} \sum_{u = -N_1}^{N_1} \sum_{v= -N_2}^{N_2}
    c_{n, rL_1+u, t L_2+v} \,
        w_n \Bigl( \frac{j}{b_n L_2} - a_n u - a_n L_1 r \Bigr)
    \e^{2 \pi i (tj + \frac{v j}{L_2})}\\[-10pt]
  &=	a_n b_n
    \sum_{u = -N_1}^{N_1} \sum_{v= -N_2}^{N_2}
    w_n \Bigl(\frac{j}{b_n L_2} - a_n u \Bigr)
    \e^{2 \pi i  \frac{v j}{L_2}}
    \sum_{r \in \Z} \sum_{t \in \Z} 
    c_{n, rL_1+u, t L_2+v} \\
  &=	a_n b_n
    \sum_{u = -N_1}^{N_1} \sum_{v= -N_2}^{N_2}
    w_n \Bigl(\frac{j}{b_n L_2} - a_n u \Bigr)
    \e^{2 \pi i  \frac{v j}{L_2}}
    Q_{n, u}(\nu) P_{n, v}(\tau),
    \label{eq:disc-data-y}
\end{align}
where in the last line we abbreviate
\begin{equation}
\sum_{r \in \Z} \sum_{t \in \Z} 
	c_{n, rL_1+u, tL_2+v}
= \underbrace{\sum_{r \in \Z} 
	\hat q_n \bigl(a_n (r L_1 + u) - \tau \bigr)
}_{\eqqcolon Q_{n, u} (\tau)}
\underbrace{\sum_{t \in \Z} 
	\hat p_n \bigl(b_n(t L_2 + v) - \nu \bigr)
}_{\eqqcolon P_{n, v} (\nu)}.
\label{eq:heckel_formula:6}
\end{equation}
We can significantly simplify \eqref{eq:heckel_formula:6} via
Poisson's summation formula: Indeed, $\hat q_n, \hat p_n$ are
band-limited, integrable functions, so by
Lemma~\ref{thm:poisson_bandlimited} we obtain
\begin{equation}
Q_{n, u}(\tau)
= \sum_{r \in \Z} 
	\hat q_n \bigl(a_n (r L_1 + u) - \tau \bigr)
= \frac{1}{a_n L_1} \sum_{k = -N_1}^{N_1}
	 q_n \Bigl(\frac{-k}{a_nL_1} \Bigr) \e^{2 \pi i \frac{k (a_n  u - \tau)}{a_n L_1} }
\end{equation}
and 
\begin{equation}
P_{n, v}(\nu)
= \sum_{t \in \Z} 
	\hat p_n \bigl(b_n(tL_2 + v) - \nu \bigr)
= \frac{1}{b_n L_2} \sum_{\ell =-N_2}^{N_2}
	 {p}_n \Bigl(\frac{-\ell}{b_nL_2} \Bigr) \e^{2 \pi i \frac{\ell(b_n  v - \tau)}{b_n L_2}}.
\end{equation}
We used that $q_n(\frac{-k}{a_n L_1}) = 0$ if $|k| \geq \frac{L_1}{2}$
since this implies
$\frac{|k|}{a_nL_1} \geq \frac{1}{2a_n} > \frac{\Omega_n}{2}$ and also
$p_n(\frac{-\ell}{b_nL_2}) = 0$ if $|\ell| \geq \frac{L_2}{2}$ because
then $\frac{|\ell|}{b_nL_2} \geq \frac{1}{2b_n} > \frac{\Tcal_n}{2}$.
\\[1em]
3.  Finally, we take limits.  By continuity of $w$ it is easy to
compute
\begin{equation}
\lim_{n \to \infty} w_n \Bigl(\frac{j}{b_n L_2} - a_n u \Bigr)
= \lim_{n \to \infty} w \Bigl(\frac{j}{\lambda_n b_n L_2} - \frac{a_n}{\lambda_n} u \Bigr)
= w \Bigl(\frac{\Tcal j}{L_2} - \frac{u}{\Omega} \Bigr).
\end{equation}
Now consider the limits of $Q_{n, u}(\tau)$ and $P_{n, v}(\nu)$.
It follows from $a_n \Omega = \lambda_n > \beta > \frac{2N_1}{L_1}$ that 
$\frac{|k|}{a_nL_1} 
\leq \frac{N_1}{a_nL_1} < \frac{\Omega}{2}$ for $k = -N_1, \dots, N_1$, 
which in turn implies $q_n(\frac{-k}{a_nL_1}) = 1$.
Similarly, since $b_n \Tcal = \lambda_n > \beta > \frac{2N_2}{L_2}$ 
we have $\frac{|\ell|}{b_nL_2} \leq \frac{N_2}{a_nL_2} 
< \frac{\Tcal}{2}$ and thus $p(\frac{-\ell}{b_nL_2}) = 1$ 
for all $\ell = -N_2, \dots, N_2$.
Consequently, it follows 
\begin{align}
  \lim_{n\to \infty} 
  Q_{n, u}(\tau)
  &= \lim_{n\to\infty}
    \frac{1}{a_n L_1} \sum_{k = -N_1}^{N_1}
    \e^{2 \pi i \frac{k (a_n  u - \tau)}{a_n L_1} }
    = \frac{\Omega}{L_1} \sum_{k = -N_1}^{N_1}
    \e^{2 \pi i \frac{k (u - \Omega \tau)}{L_1}}
    = \frac{\Omega}{L_1}D_{N_1}\left(\frac{u - \Omega \tau}{L_1}\right)
\end{align}
and by the an analogous computation,
\begin{align}
\lim_{n \to \infty} 
	P_{n, v}(\nu)
= \frac{\Tcal}{L_2} 
	D_{N_2}\left(\frac{v - \Tcal \nu}{ L_2}\right).
\end{align}
Therefore taking the limit of \eqref{eq:disc-data-y} yields
\begin{align}
&\lim_{n \to \infty} y_n \Bigl(\frac{j}{b_n L_2} \Bigr)\\
  &= \frac{1}{\Tcal \Omega}
    \sum_{u = -N_1}^{N_1} \sum_{v= -N_2}^{N_2}
    w \Bigl(\frac{\Tcal j}{L_2} - \frac{u}{\Omega} \Bigr)
    \e^{2 \pi i  \frac{v j}{L_2}}
    \frac{\Omega}{L_1} D_{N_1} \Bigl(\frac{u - \Omega \tau}{L_1} \Bigr)
    \frac{\Tcal}{L_2} D_{N_2} \Bigl(\frac{v - \Tcal \nu}{ L_2} \Bigr)\\
  &= \frac{1}{L_1 L_2}
    \sum_{u = -N_1}^{N_1} \sum_{v= -N_2}^{N_2}
    w \Bigl(\frac{\Tcal \Omega j - u L_2}{\Omega L_2} \Bigr)
    \e^{2 \pi i  \frac{v j}{L_2}}
    D_{N_1} \Bigl(\frac{u - \Omega \tau}{L_1} \Bigr)
    D_{N_2} \Bigl(\frac{v - \Tcal \nu}{ L_2} \Bigr).
\label{eq:heckel_formula:7}
\end{align}
Next we consider the limit of the definition of $y_n(\frac{j}{b_n L_2})$, i.e.,
\begin{equation}
\lim_{n \to \infty}
	y_n \Bigl(\frac{j}{b_n L_2} \Bigr)
= \lim_{n \to \infty}
		p_n \Bigl(\frac{j}{b_n L_2} \Bigr) 
		M_\nu T_\tau (\hat q_n*w_n) \Bigl(\frac{j}{b_nL_2} \Bigr).
	\label{eq:y-def-limit}
\end{equation}		
Using the assumptions on $q_n$ we obtain
\begin{align}
  \hat q_n*w_n (x)
  &= \int_\R w_n(t) \hat q_n(x - t) \intd t
    = \int_\R \int_\R \e^{-2\pi i \xi t} \intd \mu_{w_n}(\xi) \, \hat
    q_n(x - t) \intd t
  \\
  &= \int_\R \e^{-2\pi i \xi x}
    \underbrace{
    \int_\R \e^{2\pi i \xi t} \hat q_n(t) \intd t
    }_{\mathclap{=q_n(\xi) = 1 \, \text{on}\,\supp \mu_{w_n}}}
    \intd \mu_{w_n}(\xi)
    = w_n(x)
\end{align}
for all $n\in\N$, so that \eqref{eq:y-def-limit} can be written as
\begin{align}
\lim_{n \to \infty}
	y_n \Bigl(\frac{j}{b_n L_2} \Bigr)
&= \lim_{n \to \infty}
		p_n \Bigl(\frac{j}{b_n L_2} \Bigr) 
		w_n \Bigl(\frac{j}{ b_n L_2} - \tau \Bigr)
		\e^{2 \pi i \frac{j \nu}{b_n L_2}}.
\end{align}
We already showed in a previous argument that
$p_n(\frac{j}{b_n L_2}) = 1$ for $j = -N_2, \dots, N_2$ for all
$n\in\N$.  Then it follows from continuity that
\begin{align}
\lim_{n \to \infty}
	y_n \Bigl(\frac{j}{b_n L_2} \Bigr)
&= \lim_{n \to \infty}
		w  \Bigl(\frac{j}{\lambda_n b_n L_2} - \frac \tau{\lambda_n}  \Bigr)
		\e^{2 \pi i \frac{j \nu}{b_n L_2}}\\
  &= w  \Bigl(\frac{\Tcal j}{L_2} - \tau  \Bigr)
    \e^{2 \pi i \frac{\Tcal j \nu}{L_2}}
    = M_\nu T_\tau w \Bigl(\frac{\Tcal j}{L_2} \Bigr).
    \label{eq:heckel_formula:8}
\end{align}
Combining \eqref{eq:heckel_formula:7} with \eqref{eq:heckel_formula:8} then proves \eqref{eq:heckel_formula:1}.
\end{proof}

\section{Numerical Algorithms}
\label{sec:algs}
\noindent
In this section, we propose to solve problem \eqref{problem_2}, i.e.,
\begin{equation}
\argmin_{\eta \in \C^S, (\tau,\nu) \in X^S} \|G \sum_{s=1}^{S} \eta_s a(\tau_s,\nu_s) - y\|^2 
+ \lambda \|\eta\|_{\ell_1}, \qquad \lambda > 0
\end{equation}
by two kind of algorithms.  We adapt the alternating descent
conditional gradient method from \cite{BSR17} to our setting in
Subsection \ref{alg:adcg}. We will address the theoretical convergence
behaviour in a forthcoming and refer only to the literature here.  For
numerical comparisons, we start with a simple grid refinement
algorithm in the next subsection \ref{sec:alg:grid}.

\subsection{Multi-Level Time-Frequency Refinement Algorithm} 
\label{sec:alg:grid}

Instead of solving the optimization problem over the continuous set 
$X= [-\Tcal/2, \Tcal/2] \times [-\Omega/2, \Omega/2]$, 
we may discretize $X$ on a grid $\mathcal J$ of cardinality $J$. 
For instance we could choose an equidistant grid.
Then we consider the atoms on the grid points $(\tau_j, \nu_j)$, $j \in \mathcal J$.
Setting
$$Z_{\mathcal J} \coloneqq [a(\tau_1,\nu_1), \dots, a(\tau_J, \nu_J)] \in \mathbb C^{L_1  L_2 \times J}$$ 
and 
$\eta \in \mathbb C^J$, 
we reduce 
\eqref{eq:tv-tik}
to the convex minimization problem
\begin{equation} \label{eq:anp_bp}
\min_{\eta \in \C^{J}} \| G Z_{\mathcal J} \eta - y \|^2 + \lambda \, \|\eta\|_1.
\end{equation}
The sparsity of the discrete measure is here promoted by the 1-norm.
In other words, we hope that $\eta$ has only $S \ll J$ entries which
are not near zero.  For one-dimensional problems on the torus, Duval
and Peyr\'e \cite{DP17} showed that the discretized problem
$\Gamma$-converges to the continuous problem in the sense of
Remark~\ref{rem:cont} if the regular grid gets finer and finer under
certain assumptions; so if the grid is fine enough, we should obtain a
sufficient precise solution.  On the contrary, a fine grid blows up
the problem dimension and make its numerically intractable.  Further,
as described in \cite{DDPS20} and references therein for general total
variation minimization problems, the true point masses are usually
approximated by several point masses of the grid in a small
neighbourhood.  These clusters may be detected and replaced by an
averaged point mass.  Further, the minimization problem
\eqref{eq:anp_bp} is a basis pursuit often encountered in compressed
sensing and can be solved using toolboxes like \texttt{CVX} \cite{cvx}
or, approximately, by greedy methods like matching pursuits
\cite{DMA97,Tro04,CW11}.

\begin{algorithm}\caption{Orthogonal matching pursuit}\label{alg:gmp}
  \begin{algorithmic}[1]
    \Require $y$, ${\mathcal J}$.
    \State Set $r := y$,
    $\tau := [\,]$, $\nu := [\,]$, $Z := [\,]$.
    \For{$k = 0, 1, 2, \dots$}
    \State Expand $\tau \in \mathbb R^k$,
    $\nu \in \mathbb R^k$ by
    $$
    (\tau_{k+1}, \nu_{k+1}) := \argmax_{(\tau, \nu) \in {\mathcal J}}
    \tfrac{|\langle r, G a(\tau, \nu)\rangle|}{\|G a
      (\tau, \nu) \|_2}.
    $$
    \State Expand $Z \in \mathbb C^{L_1L_2 \times k}$ by
    $a(\tau_{k+1}, \nu_{k+1})$.
    \State Compute the least-square solution of
    $$\min_{\eta \in
      \mathbb C^{k+1}} \|GZ \eta - y\|_2^2.$$
    \State Set $r := y - GZ \eta$.
    \EndFor
    \State \Return $(\tau, \nu, \eta)$.
  \end{algorithmic}
\end{algorithm}

Instead of choosing a fine grid on the entire domain, we would like to
solve the $\ell^1$ minimization problem \eqref{eq:anp_bp} on a small
set $\mathcal J$ that, in the ideal case, only covers the
neighbourhoods of the unknown true parameters in $X$ to reduce the
numerical effort.  For this purpose, we initially apply the orthogonal
matching pursuit in Algorithm~\ref{alg:gmp} on a fine regular grid
until the residuum $r$ gets small or a certain number of atoms is
determined.  Although the performance of the greedy method strongly
depends on the current instance, the computed atoms are usually
located near the true point masses.  Surrounding the computed atoms
with a fine local grid, we obtain a good starting set ${\mathcal J}_0$
for \eqref{eq:anp_bp}.  Next, we would like to let the local grid
become finer and finer to improve the solution and to let the number
of atoms be nearly the same.  Having an optimal $\eta^*$ of
\eqref{eq:anp_bp} for ${\mathcal J}_r$, we may chose a new finer grid
${\mathcal J}_{r+1}$ around the interesting features by one of the
following refinement strategies:
\begin{enumerate}
\item Determine the dominant atoms corresponding to
  $(\tau_j, \nu_j) \in {\mathcal J}_r$ with $|\eta_j^*| \ge \epsilon$.  Discretize
  the neighbourhood around these atoms by a finer grid.  Chose
  ${\mathcal J}_{r+1}$ as the union of these finer grids.
\item Determine the importance $\gamma_j$ of the atom corresponding to
  $(\tau_j, \nu_j) \in {\mathcal J}_r$ by
  \begin{equation}
    \gamma_j := \sum_{(\tau_k, \nu_k) \in {\mathcal J}_r \cap U_j} |\eta_k^*|,    
  \end{equation}
  where the coefficients of all atoms with parameters in a
  neighbourhood $U_j$ around $(\tau_j, \nu_j)$ are summed up.  For the
  most important neighbourhood $U_j$, compute the barycenter by
  \begin{equation}
    (\tilde\tau_j, \tilde\nu_j) :=
    \sum_{(\tau_k, \nu_k) \in {\mathcal J}_r \cap U_j} \tfrac{|\eta_k^*|}{\gamma_j} \,
    (\tau_k, \nu_k).
  \end{equation}
  Add a finer grid around $(\tilde\tau_j, \tilde\nu_j)$ to ${\mathcal J}_{r+1}$,
  remove the atoms in $U_j$ from ${\mathcal J}_r$, and repeat the procedure as
  long as there are important points with $\gamma_j \ge \epsilon$.
\end{enumerate}
The new local grids should cover a smaller neighbourhood.  For
instance, these grids could again be regular with decreasing step size
according to $r$.  Notice that the numerical effort of the first
refinement strategy is less than for the second one.  On the other
hand, the second strategy can leave the local grids due to the
barycenters.  After determining a final atomic set ${\mathcal J}^*$ containing
the most dominant atoms or barycenters, the corresponding coefficients
can be computed by solving the least square problem
\begin{equation}
  \label{eq:least-square}
  \min_{\eta \in \C^{|{\mathcal J}^*|}} \| G Z_{{\mathcal J}^*} \eta - y \|^2.  
\end{equation}
In summary, we obtain Algorithm~\ref{alg:mp_ml}.

\begin{algorithm}\caption{Multi-level time-frequency
    refinement}\label{alg:mp_ml} 
  \begin{algorithmic}[1]
    \Require $y$. 
    \State Construct an initial grid ${\mathcal J}_0$ using Algorithm~\ref{alg:gmp}.
    \State Compute the minimizer $\eta^*$ of \eqref{eq:anp_bp}.
    \For{$r = 0, 1, 2, \dots$} 
    \State Determine a new atomic set ${\mathcal J}_{r+1}$ using strategy 1 or 2.
    \State Compute the minimizer $\eta^*$ of \eqref{eq:anp_bp}.
    \EndFor
    \State Determine the dominant atoms or centers as in strategy 1 or
    2.
    \State Compute $\eta$ by solving \eqref{eq:least-square}.
    \State \Return $(\tau, \nu, \eta)$.
  \end{algorithmic}
\end{algorithm}

\subsection{Alternating Descent Conditional Gradient Algorithm}
\label{alg:adcg}
Next, we adapt the\linebreak ADCG from \cite{BSR17} to our setting.
This algorithm minimizes over the continuous domain $X$.  The ADCG is
a modification of the conditional gradient method (CGM) -- also known
as the Frank-Wolfe algorithm introduced in \cite{FW56} -- for total
variation regularization.  The original Frank-Wolfe algorithm on
$\mathbb R^d$ solves optimization problems of the form
$\argmin_{x \in \mathcal V} f(x)$, where the feasible set
$\mathcal V \subset \R^d$ is compact and convex and the function $f$
is a differentiable and convex.  Given the $k$th iterate $x_k$ each
iteration consists basically of two steps, namely
\begin{itemize}
\item[i)] minimizing a linearized version of $f$ in $x_k$ over the feasible set
$$v_k = \argmin_{v \in \mathcal V} f(x_k) + \langle \nabla f(x_k), v - x_k\rangle,$$
\item[ii)] updating with 
$$x_{k+1} = x_k +  \gamma (v_k - x_k)$$.
\end{itemize}
In superresolution, the first step always consists in an update of the support of the measure
as it is also done in the first step of our Algorithm \ref{alg:mp_gd}.

Concerning the second step, all important convergence guarantees of the algorithms are still valid, if
we replace $x_{k+1}$ in the second step by any feasible $\tilde x_{k+1}$ that fulfills
$f(\tilde x_{k+1}) \le f(x_{k+1})$. 
This flexibility has led to several successful variations of the classical Frank-Wolfe algorithm. 
ADCG related algorithms which differ in the second step are for example 
the algorithm in \cite{BP13} 
and the so-called sliding Frank-Wolfe in \cite{DDPS20}. 
While the first one uses soft shrinkage to update the amplitudes and a discrete gradient flow over the locations,
the second one uses a non-convex solver to jointly minimize over the amplitudes and positions
with a suitable starting values for the  amplitudes.
 
Adapting the ADCG to our setting results in Algorithm~\ref{alg:mp_gd}, whose
details are discussed in the following.  For convergence results we refer to \cite{BSR17}.
\begin{algorithm}\caption{ADCG}\label{alg:mp_gd} 
  \begin{algorithmic}[1]
    \Require $y$. 
    \State Set $r := y$, $\tau := [\,]$, $\nu := [\,]$.
    \For{$k = 0,1,2, \dots$}
    \State Expand $\tau \in \R^{J_k}$, $\nu \in \R^{J_k}$ by
    $$
    (\tau_{{J_k}+1}, \nu_{{J_k}+1}) := \argmax_{(\tau, \nu) \in \Ucal}
    |\langle r, G a(\tau, \nu)\rangle|.
    $$
    \For{$\ell = 0,1,2,\dots$}
    \State Compute a minimizer
    $$
    \eta := \argmin_{\eta \in \mathbb C^{J_k+1}} \| GZ(\tau, \nu) \eta -
    y \|_2^2 + \lambda \| \eta \|_1.
    $$
    \State Compute a minimizer
    $$
    (\tau, \nu) := \argmin_{(\tau,\nu) \in \Ucal^{J_k+1}}
    \|GZ(\tau,\nu) \eta - y \|_2^2
    $$
    \EndFor
    \State Remove point masses with zero coefficients.
    \State Set $r := y - G Z(\tau,\nu) \eta$.
    \EndFor
    \State \Return $(\eta, \tau, \nu)$.
  \end{algorithmic}
\end{algorithm}
The expansion step of the ADCG algorithm is very
similar to the greedy matching pursuit in Algorithm~\ref{alg:gmp}
without normalization of the atoms.  To find a solution 
\begin{equation}
  (\tau_{{J_k}+1}, \nu_{{J_k}+1}) := \argmax_{(\tau, \nu) \in X}
  |\langle r, G a(\tau, \nu)\rangle|,
\end{equation}
the objective can first be evaluated on a fine regular grid of
$X$.  The obtained $(\tau_{{J_k}+1}, \nu_{{J_k}+1})$ may then be
improved using a gradient descent method.  In our numerical
simulations, we however notice that this improvement step has no
crucial impact on the recovered measure for our problem and can be
skipped.

The second step consists in the update of the parameters by
\begin{equation}
  (\eta, \tau, \nu) := \argmin_{\eta \in \mathbb C^{J_k+1},
    (\tau,\nu) \in X^{J_k+1}} \| GZ(\tau, \nu) \eta -
  y \|^2 + \lambda \| \eta \|_1
\end{equation}
with 
$Z (\tau, \nu) \coloneqq [a(\tau_1,\nu_1),\dots, a(\tau_S, \nu_S)]$.  
In
difference to the general algorithm in \cite{DDPS20}, the coefficient
of the point masses $\eta$ are complex numbers such that the above
update consists in the minimization of a non-smooth objective.
Therefore, we use the alternating minimization proposed in
\cite{BSR17}, which splits up the minimization into the basis pursuit
or LASSO problem
\begin{equation}  
  \eta := \argmin_{\eta \in \mathbb C^{J_k+1}} \| GZ(\tau, \nu) \eta -
  y \|^2 + \lambda \| \eta \|_1.
\end{equation}
and the smooth minimization problem
\begin{equation}
  (\tau, \nu) := \argmin_{(\tau,\nu) \in X^{J_k+1}}
  \underbrace{\|GZ(\tau,\nu) \eta - y \|_2^2}_{=:F(\tau,\nu)}.    
\end{equation}

The $\ell^1$ regularized problem can be solved as discussed above and the second
one by a gradient descent or quasi Newton method like BFGS.  A short
computation shows that the gradients of the objective $F$ are just
given by
\begin{align}
    \grad_\tau F(\eta,\tau,\nu)
  &=
    2 \real \bigl\{ \left(G Z^\tau(\tau,\nu) \diag(\eta) \right)^*
    (G Z(\tau,\nu) \eta - y) \bigr\},
  \\
  \grad_\nu F(\eta,\tau,\nu)
  &=
    2 \real \bigl\{ \left(G Z^\nu(\tau,\nu) \diag(\eta) \right)^*
    (G Z(\tau,\nu) \eta - y) \bigr\},
\end{align}
where $\cdot^*$ denotes the conjugation and transposition of a matrix.
The partial derivatives of the atoms $a(\tau_j,\nu_j)$ with respect to
$\tau_j$ and $\nu_j$ are collected in the matrices
\begin{align}
  Z^\tau(\tau,\nu)
  &:= [\tfrac{\mathrm{d}}{\mathrm{d} \tau} a(\tau_1, \nu_1),
    \dots,
    \tfrac{\mathrm{d}}{\mathrm{d} \tau} a(\tau_s, \nu_s)],
  \\
  Z^\nu(\tau,\nu)
  &:= [\tfrac{\mathrm{d}}{\mathrm{d} \nu} a(\tau_1, \nu_1),
    \dots,
    \tfrac{\mathrm{d}}{\mathrm{d} \nu} a(\tau_s, \nu_s)]
\end{align}
with
\begin{align}
  [\tfrac{\mathrm{d}}{\mathrm{d} \tau} a(\tau, \nu)]_{u,v}
  &= - \tfrac{\Omega}{L_1^2 L_2} \, D'_{N_1} \Bigl(
  \frac{u-\Omega\tau}{L_1} \Bigr) \, D_{N_2} \Bigl( \frac{v - \Tcal
    \nu}{L_2} \Bigr),
  \\
  [\tfrac{\mathrm{d}}{\mathrm{d} \nu} a(\tau, \nu)]_{u,v}
  &= - \tfrac{\Tcal}{L_1 L_2^2} \, D_{N_1} \Bigl(
  \frac{u-\Omega\tau}{L_1} \Bigr) \, D'_{N_2} \Bigl( \frac{v - \Tcal
    \nu}{L_2} \Bigr).
\end{align}
The derivative of the $N$-th  Dirichlet kernel $D_N$ is given by
\begin{equation}
  D'_N(x) = - 4 \pi \sum_{k=1}^N k \sin (2 \pi k x)
  = \biggl( \frac{\sin((2N+1) \pi x)}{\sin(\pi x)} \biggr)'.  
\end{equation}
Finally, we like to mention that the numerical effort of ADCG algorithm
is much higher compared with the multi-level refinement in
Algorithm~\ref{alg:mp_ml} since several optimization problems have to
be solved for each added point mass.

\section{Numerical Results}
\label{sec:numerics} 
\noindent
In the following experiments, we compare the orthogonal matching
pursuit, the multi-level time-frequency refinement, and the ADCG.  
First, we consider the performance for a
specific synthetic instance.  Then we study the general
performance with respect to the noise level and how many measurements
are needed to estimate the unknown channel.  Finally, the influence of
the identifier model is discussed.

{\bfseries Channel estimation from synthetic measurements.}
For this experiment, we assume that the unknown channel or operator
$H$ in \eqref{operator} has exactly $s=10$ features and that this
number is known in advance.  The shifts and modulations
$(\tau_j, \nu_j)$ are independently generated with respect to the
uniform distribution on
$[-\Tcal / 2, \Tcal / 2] \times [-\Omega/2, \Omega/2] = [-1.5, 1.5]
\times [15.5, 15.5]$.  The coefficients $\eta_j$ are independently and
uniformly drawn from the complex unit circle.  The employed identifier
$w$ is a trigonometric polynomial of degree $N_1 = 50$, i.e.\
$L_1 = 101$, whose coefficients are independently drawn from the
complex unit circle too.  The true samples
$y_j = H w (\tfrac{\Tcal j}{L_2})$ with $j = -N_2, \dots, N_2$ and
$L_2 = 101$ are corrupted by additive complex Gaussian noise such that
$|| y - y^\delta || / || y || = 0.1$, which corresponds to
$-10~\mathrm{db}$\footnote{The unit decibel henceforth refers to the
  scale $10 \log_{10}(\| \cdot - p\| / \| p \|)$ for a reference point
  $p$ -- usually the true measurements or operator.  Depending on the
  context, the norm refers to the Euclidean or operator norm.}  white
noise -- the noisy data are again denoted by $y^\delta$.

To recover the unknown channel parameters, we apply the orthogonal
matching pursuit (Algorithm~\ref{alg:gmp}) with the regular grid $S$
of $[-\Tcal/2, \Tcal/2] \times [-\Omega/2, \Omega/2]$ consisting of
1\,024 points in each direction.  The same grid is used to compute the
location of the new point masses in the ADCG (Algorithm~\ref{alg:mp_gd}).  
Both methods are stopped after computing
exactly 10 features.  The multi-level refinement in
Algorithm~\ref{alg:mp_ml} is initialized by applying the orthogonal
matching pursuit to a coarser grid with 256 points in each direction.
The local $5\times 5$ grids are then refined 15 times by reducing the
stepsize by a factor of $0.75$.  
We always use the second refinement strategy. 
The multi-level refinement and the
ADCG are applied to the Tikhonov regularization
\eqref{eq:tv-tik} with $\lambda = 500$.  The recovered shifts and
modulations of all three methods are shown in
Figure~\ref{fig:exact-atom-num}.  The true parameters are denoted with
an additional $\dagger$.  The absolute errors of the estimation are
recorded in Table~\ref{tab:exact-atom-num}, where the experiment has
been repeated 50 times and the errors are averaged.  For this instance, all
three methods yield comparable results, where the shifts $\tau_j$ and
modulations $\nu_j$ are quite accurate.
The multi-level refinement and
the ADCG method achieve slightly higher accuracies
than the orthogonal matching pursuit, but, on the downside, the
ADCG method is much more time-consuming than the others.
Considering the
noise level, the results are nevertheless satisfying and show that in
particular the shifts and modulations are recoverable from highly
noisy measurements.

\begin{figure}
  \centering
  \begin{subfigure}[c]{0.3\linewidth}
    \includegraphics[width=\linewidth]{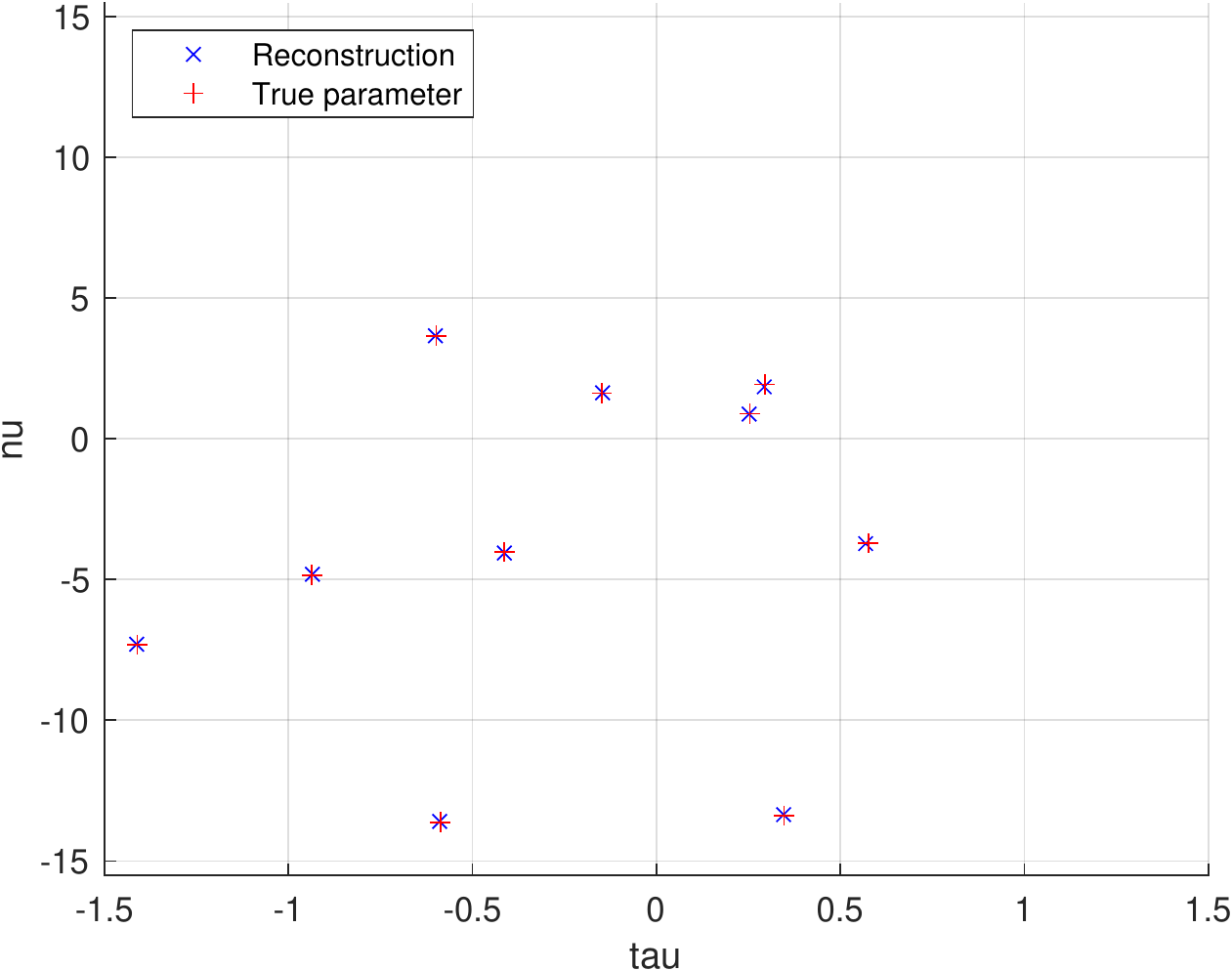}
    \subcaption{Orthogonal matching pursuit (Algorithm~\ref{alg:gmp}).}
  \end{subfigure}
  \quad
  \begin{subfigure}[c]{0.3\linewidth}
    \includegraphics[width=\linewidth]{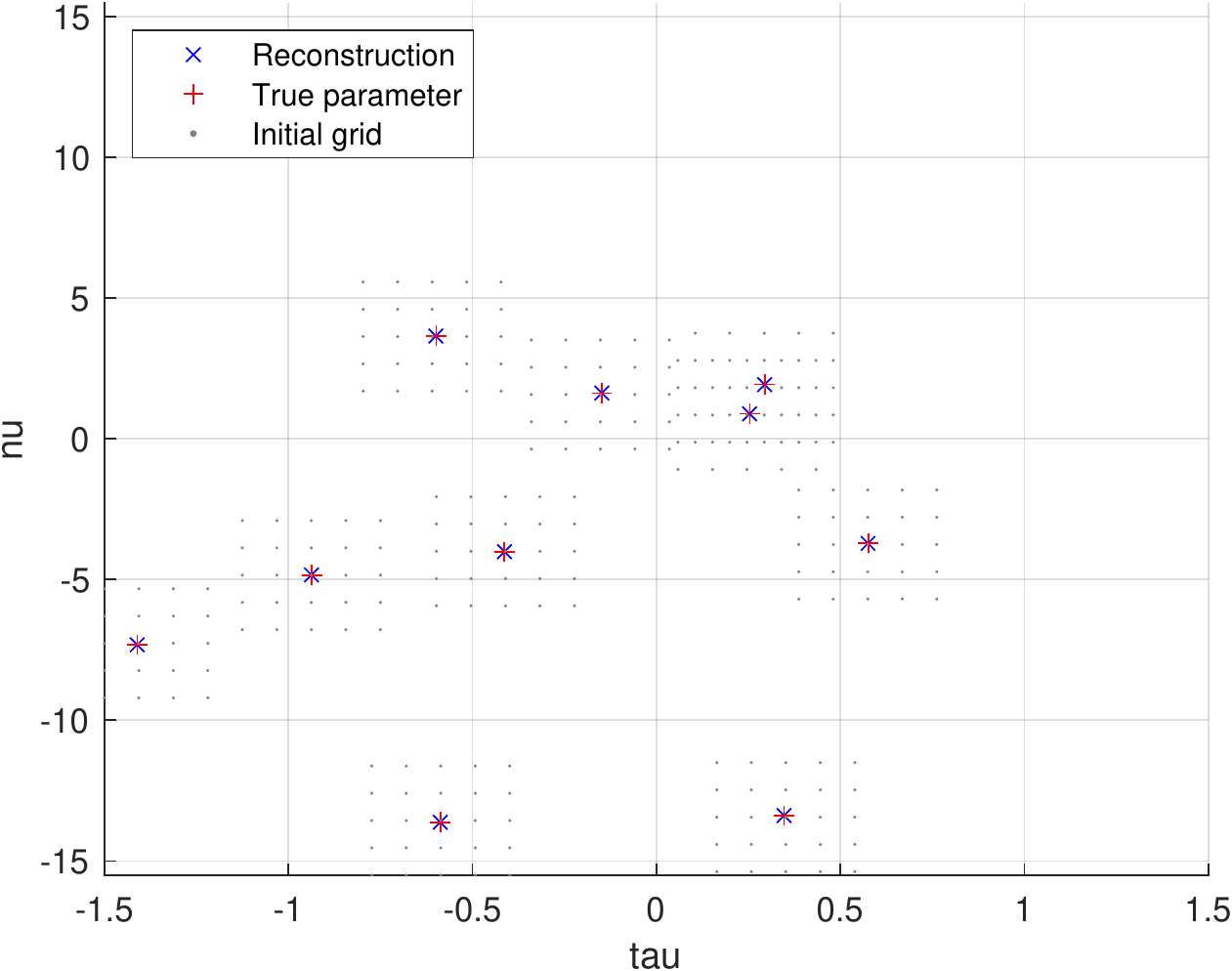}
    \subcaption{Multi-level grid refinement (Algorithm~\ref{alg:mp_ml}).}
  \end{subfigure}
  \quad
  \begin{subfigure}[c]{0.3\linewidth}
    \includegraphics[width=\linewidth]{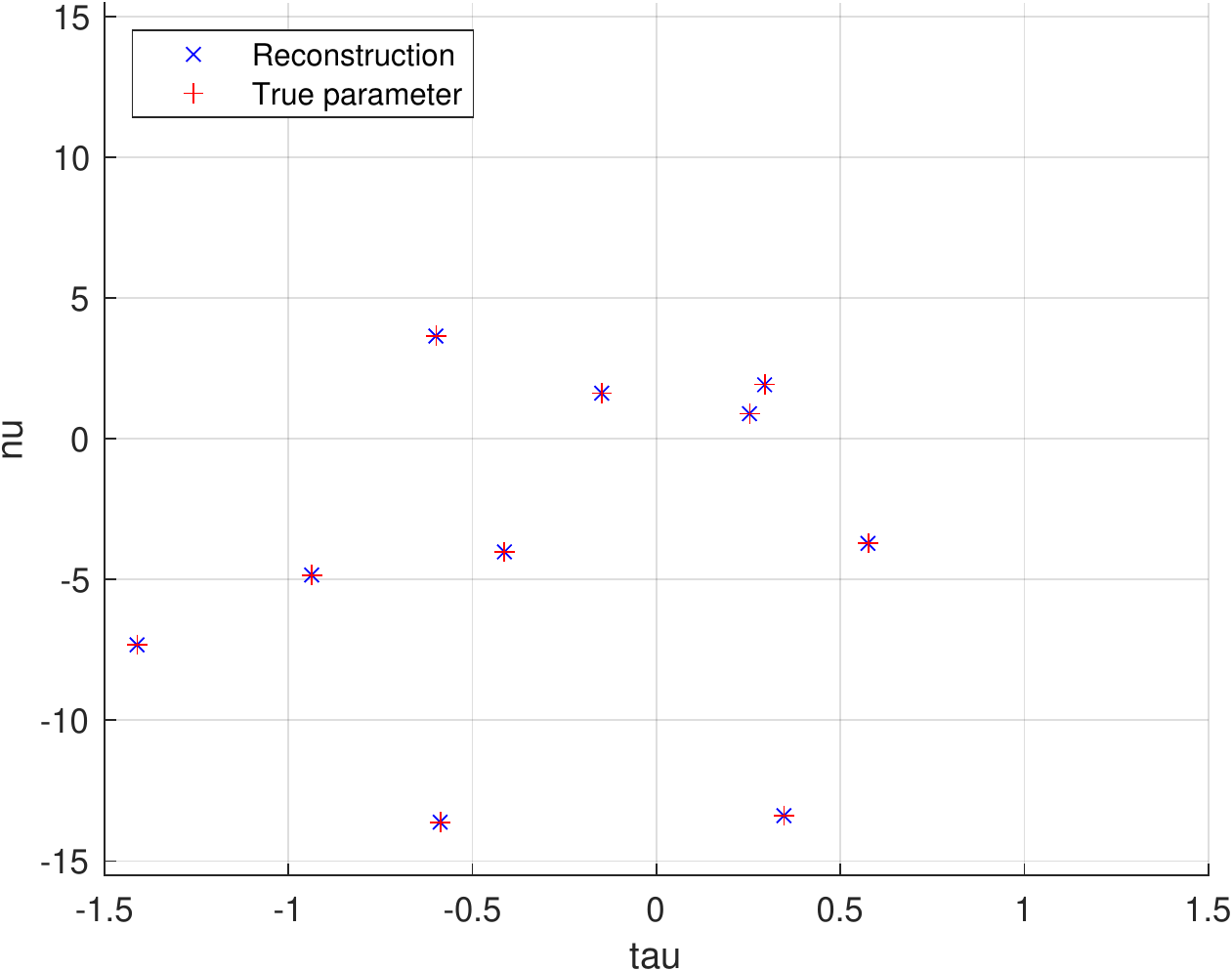}
    \subcaption{ADCG (Algorithm~\ref{alg:mp_gd}).\\ {\ }}
  \end{subfigure}
  \caption{Estimated shifts and modulations of a channel with $s=10$
    features, where $s$ is exactly known in advance.  The degree of
    the identifier and the number of samples are $L_1 = L_2 = 101$.
    The additive Gaussian noise corresponds to $-10~\mathrm{db}$.}
  \label{fig:exact-atom-num}
\end{figure}

\begin{table}
  \centering
  \begin{tabular}[c]{|l|rrr|}
    \hline
    Absolute error
    & Algorithm~\ref{alg:gmp}
    & Algorithm~\ref{alg:mp_ml}
    & Algorithm~\ref{alg:mp_gd}
    \\
    \hline
    $\max |\tau_j^\dagger - \tau_j |$
    & $4.778 \cdot 10^{-2}$
    & $2.809 \cdot 10^{-2}$
    & $3.336 \cdot 10^{-2}$
    \\
    $\max |\nu_j^\dagger - \nu_j|$
    & $1.995 \cdot 10^{-1}$
    & $9.476 \cdot 10^{-2}$
    & $1.331 \cdot 10^{-1}$
    \\
    $\max |\eta_j^\dagger - \eta_j|$
    & $2.322 \cdot 10^{-1}$
    & $1.491 \cdot 10^{-1}$
    & $1.508 \cdot 10^{-1}$
    \\
    \hline
    Mean run time
    & 194 seconds
    & 45 seconds
    & 404 seconds
    \\
    \hline
  \end{tabular}
  \caption{Mean absolute reconstruction errors over 50 experiments for
    channels with $s=10$ features, where $s$ is exactly known in
    advance.  The degree of the identifier and the number of samples
    are $L_1 = L_2 = 101$.  The additive Gaussian noise corresponds to
    $-10~\mathrm{db}$.}
  \label{tab:exact-atom-num}
\end{table}

{\bfseries Influence of Noise.}
Next, we study the influence of the noise on the recovery quality of
the algorithms in more details.  Therefore, the unknown channel is
again randomly generated with respect to 10 coefficients on the
complex unit circle.  In contrast to the first numerical example, the
algorithms are henceforth stopped if the residuum becomes small or if
the objective stagnates; in other words, the algorithms have no
knowledge of the true sparsity $S$.  The degree of the random identifier
with unimodular coefficients and the number samples is
$L_1 = L_2 = 101$ once more.  The remaining parameters are $T = 1$ and
$\Omega = 101$.  The parameter $\lambda$ is chosen with respect to the
noise level and goes to zero for vanishing noise.  Differently from
the experiment before, we want to measure how well the estimated
channel approximates the true one.  Since we are only interested on
the behavior of the true channel on the sampled interval
$[-\Tcal/2, \Tcal/2]$, we interpret the restriction of $H$ as an
operator from the space of $L_1/\Omega$-periodic trigonometric
polynomials
$\mathcal P_{N_1} \subset L^2([- L_1/2\Omega, L_1/2\Omega))$ of degree
$N_1$ at the most to the square-integrable functions
$L^2([-\Tcal/2, \Tcal/2])$, i.e.
\begin{equation}
  H : \mathcal P_{N_1}
  \to L^2([-\tfrac{\Tcal}2, \tfrac{\Tcal}2]).
\end{equation}
The difference between the true operator $H^\dagger$ and the estimated
operator $H$ is henceforth measured by the operator norm
\begin{equation}
  \|H^\dagger - H\| :=
  \sup_{w \in \mathcal P_{N_1}\setminus\{0\}}
  \frac{\|H^\dagger w - H w\|_{L^2([-\Tcal/2,
      \Tcal/2])}}{\|w\|_{L^2([- L_1/2\Omega,
      L_1/2\Omega))}} .
\end{equation}
Due to Parseval's identity, the considered subspace is isometrically
isomorph to the coefficient space $\C^{L_1}$.  After discretizing
$[-\Tcal/2, \Tcal/2]$ and employing the midpoint rule, the
operator norm may be computed numerically using the singular value
decomposition.

The mean performance of the discussed algorithms is shown in
Figure~\ref{fig:comp-alg}, where for every noise level the experiment
has been repeated 50 times.  During the multi-level refinement, the
step size of the local grids is decreased 25 times by a factor of 2/3.
For the ADCG method the $\ell^1$ and least-square
minimization is alternated 25 times. The observation of the first
experiment for $-10\,\mathrm{dB}$ noise carry over.  Notice that
already small parameter errors lead to large relative errors in the
operator norm.  The reconstruction error for the multi-level method
and the ADCG method corresponds nearly one-to-one to
the noise level of the measurements.  The reconstruction by the
orthogonal matching pursuit does not improve if the noise is
decreasing. Although the orthogonal matching pursuit yields sufficient
results as starting point for the refinement method, the problem
cannot be solved sufficiently accurate by applying only this greedy
method.

\begin{figure}
  \centering
  \includegraphics[width=0.5\linewidth]{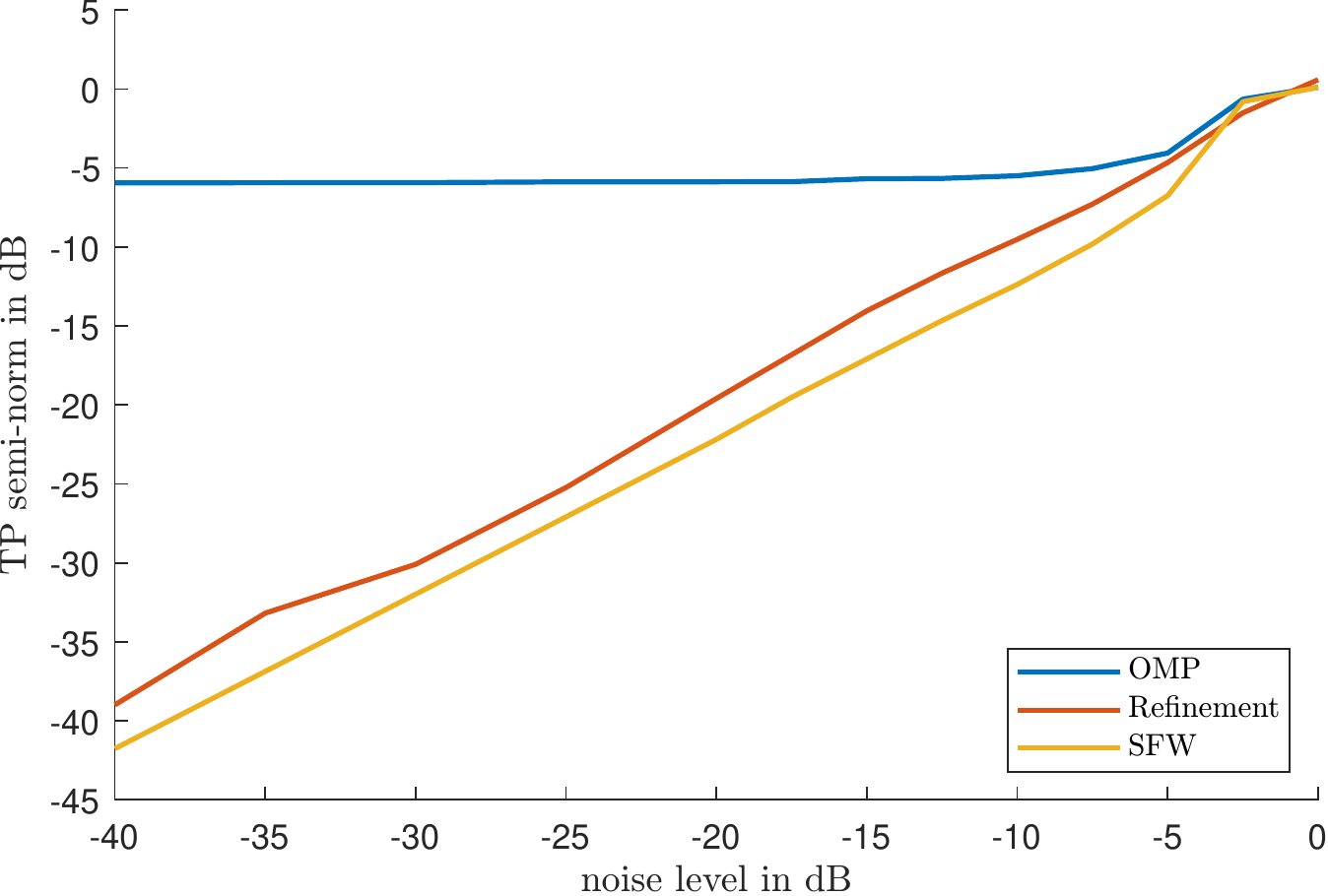}
  \caption{The recovery error of the orthogonal matching pursuit
    (Algorithm~\ref{alg:gmp}), the refinement strategy
    (Algorithm~\ref{alg:mp_ml}), and the Frank--Wolfe method
    (Algorithm~\ref{alg:mp_gd}) over varying levels of complex
    Gaussian noise to recover a channel with 10 features from 101
    samples.  The regularization parameter $\lambda$ has been chosen
    with respect to the current noise level.}
  \label{fig:comp-alg}
\end{figure}

\begin{figure}
  \centering
  \begin{subfigure}[c]{0.4\linewidth}
    \includegraphics[width=\linewidth]{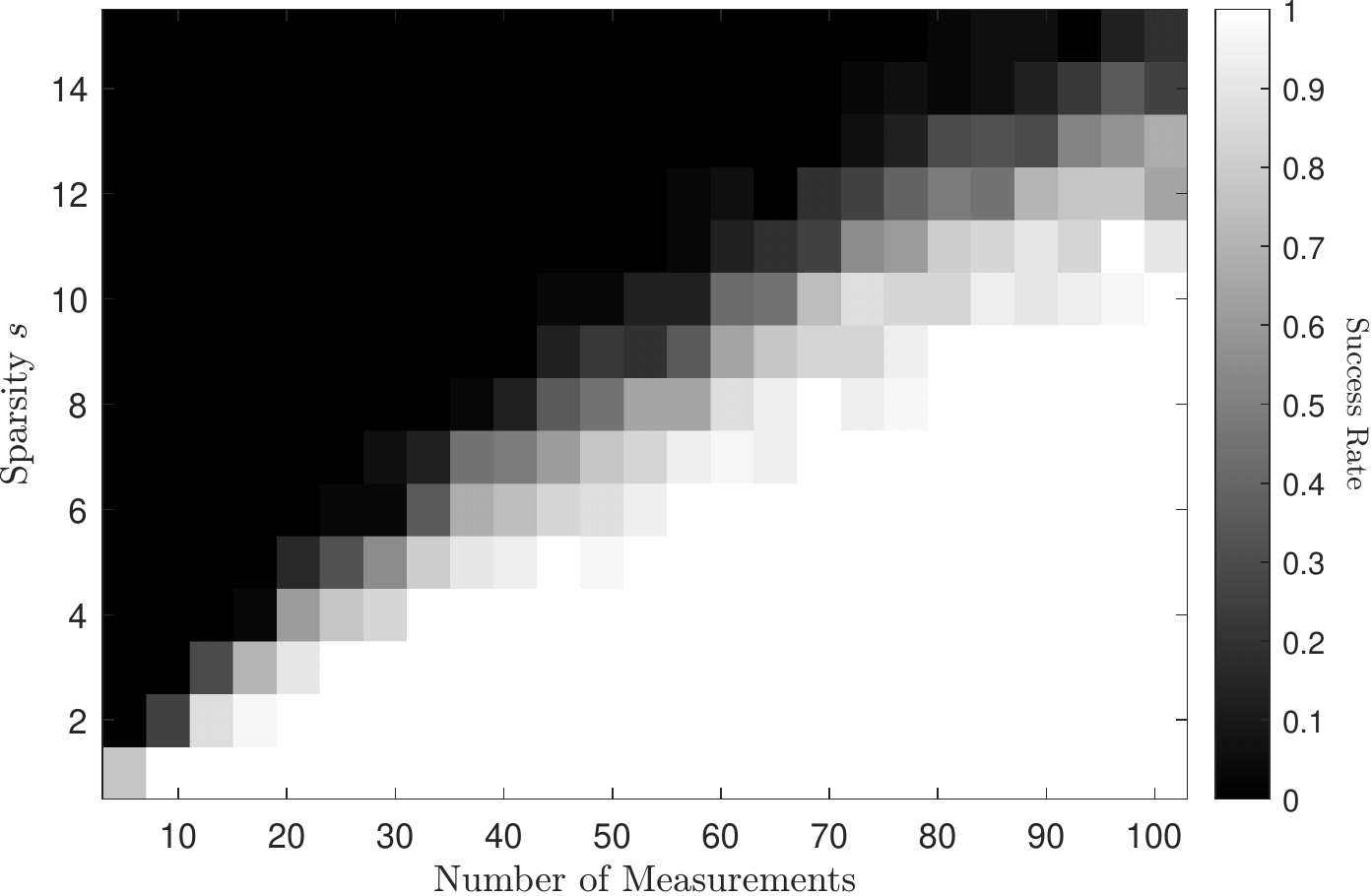}
    \subcaption{Empirical success rate.}
  \end{subfigure}
  \qquad
  \begin{subfigure}[c]{0.4\linewidth}
    \includegraphics[width=\linewidth]{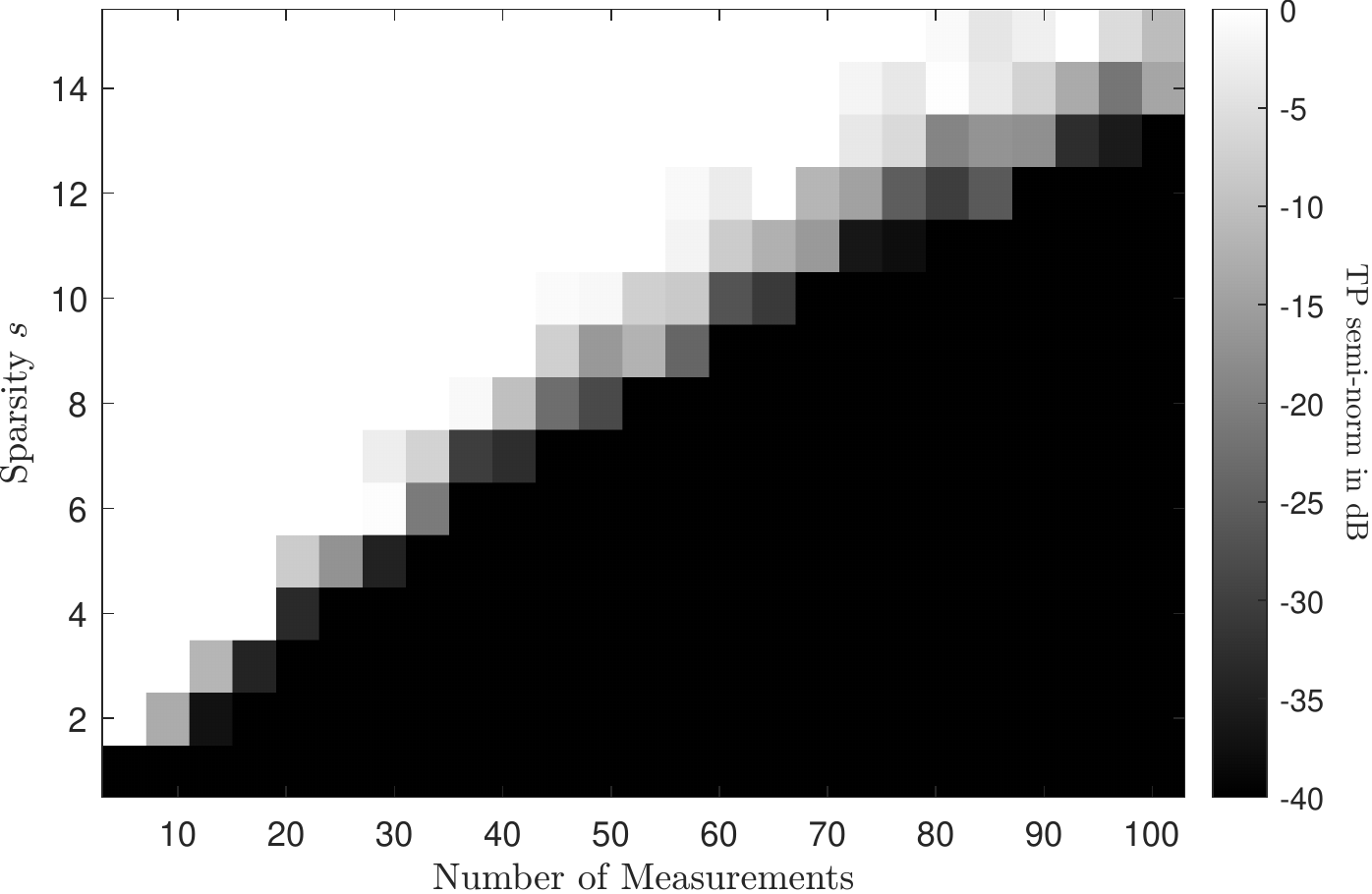}
    \subcaption{Mean error in operator norm.}
  \end{subfigure}
  \caption{Empirical probability that the operator norm satisfies
    $||H^\dagger - H|| / ||H^\dagger|| \le -40 \, \mathrm{dB}$
    depending on the number of measurements and features.  The
    coefficients of the unknown operators have been chosen unimodular.}
  \label{fig:phase-trans}
\end{figure}

{\bfseries Number of required measurements}
During our numerical experiments, we have noticed that around 10 times
more samples than unknown features are required to estimate the
parameters of the channel sufficiently well.  In the following, we
explore the question how many measurements are needed in more details.
For this, we consider the solution of Algorithm~\ref{alg:mp_gd} for
different numbers of features and numbers of measurements.  The
remaining parameters of the setting are $\Omega = L_1 = L_2 = 101$ and
$\Tcal = 1$.  The coefficient of the unknown channel are unimodular,
and the measurements are exact.  We declare a reconstruction as
success if the relative error $||H^\dagger - H||/||H^\dagger$ is less than
$-40 \, \mathrm{dB}$, and repeat the experiment 50 times for each data
point.  The success rate and the mean relative error in the operator
norm are shown in Figure~\ref{fig:phase-trans} and sustain our
observation.

This experiment is the numerical analogon to the theoretical recovery
guarantee in \cite[Thm~1]{HMS14}, where the unknown parameters
$(\eta_s, \tau_s, \nu_s)$ of \eqref{eq:derive_heckels_formula:2} in
Remark~\ref{rem:heckel} are determined by solving an atomic norm
problem.  More precisely, the minimizer of the atomic norm problem
yields the wanted parameters with high probability under certain
assumptions.  For the theoretical statement, at least $L \ge 1024$
measurements are required.  Considering the phase transition in
Figure~\ref{fig:phase-trans}, we see that, from a numerical point of
view, much less measurements are needed to recover the unknown
channel.  In particular for higher sparsity levels, the transition
between failure and success becomes non-linear, which corresponds to
the theoretical results.

{\bfseries Influence of the minimal separation.}
Continuing the discussion of the theoretical guarantees, we recall
that one of the crucial assumptions is a lower bound for the minimal
separation
\begin{equation}
  \min \bigl\{ \tfrac{|\tau_j - \tau_k|}{\Tcal},
  \tfrac{|\nu_j - \nu_k|}{\Omega} \bigr\}.
\end{equation}
If the distance between two or more features in the parameter space
become to close, they cannot be resolved numerically and are often
combined into one feature.  This well-known effect may heavily lower
the quality of the reconstruction and also occur in our setting.  To
study this behaviour numerically, we again consider random channels
with 10 unimodular features for $L_1 = L_2 = 101$, $\Tcal = 1$,
$\Omega = 101$.  The shifts and modulations are generated such that
the parameter set exactly possesses a certain minimal separation.  The
results with respect to the operator norm on $\mathcal P_{N_1}$ are
shown in Figure~\ref{fig:min-sep}, where the experiments have been
repeated 50 times without noise.  If the separation falls below
$0.01$, then the error increases rapidly.  Note that this transition
point depends on the problem dimension $L_1$, $L_2$ and on the number
of unknown features.

\begin{figure}
  \centering
  \includegraphics[width=0.5\linewidth]{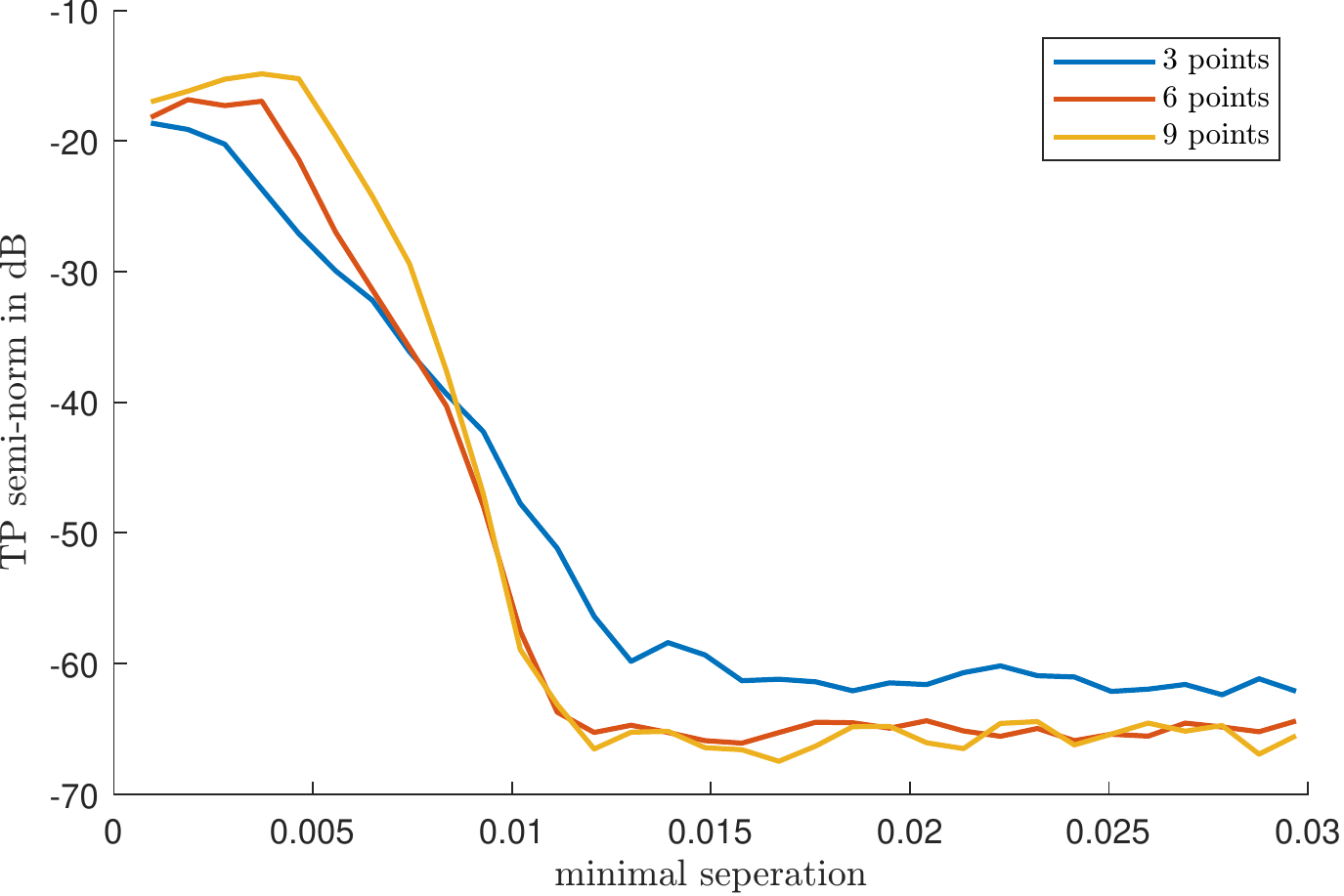}
  \caption{The recovery error for unknown channels with certain
    relative minimal separation between the 10 features.  The employed
    101 measurements have been free of noise.}
  \label{fig:min-sep}
\end{figure}

{\bfseries Importance of the identifier model.}
Finally, we study how the chosen identifier model affects the recovery
quality.  During the entire paper, we used trigonometric polynomials
as identifier $w$ for the unknown channel.  On the basis of $w$, the
given samples $Hw(\Tcal d/L_2)$ are related to the unknown parameters
by Theorem~\ref{thm:heckel_formula}, which are then determined by
solving the Tikhonov functional \eqref{eq:tv-tik} with respect to the
total variation norm for measures.  In \cite{HMS14}, for the special
case $L := L_1 = L_2$, $N := N_1 = N_2$, and odd $L = \Tcal \Omega$,
Heckel, Morgenshtern, and Soltanolkotabi have suggested to solve an
atomic norm problem based on a model approximation where the
identifier is chosen as sum of shifted sinc functions
\begin{equation}
  \label{eq:sinc-sum}
  w(x) =
  \sum_{k = -L -N}^{L + N} c_k \sinc(x \Omega - k).
\end{equation}
The real coefficients are chosen partially periodic as
$c_k = c_{k+L} = c_{k-L}$ for $k = -N, \dots, N$.  We denote the
$L$-dimensional span of the sinc functions \eqref{eq:sinc-sum} by
$\mathcal S_L$.  The given samples are then only approximated by
\eqref{eq:heckel_formula:1} in Theorem~\ref{thm:heckel_formula}.

\begin{figure}
  \centering
  \begin{subfigure}[c]{0.4\linewidth}
  \includegraphics[width=\linewidth]{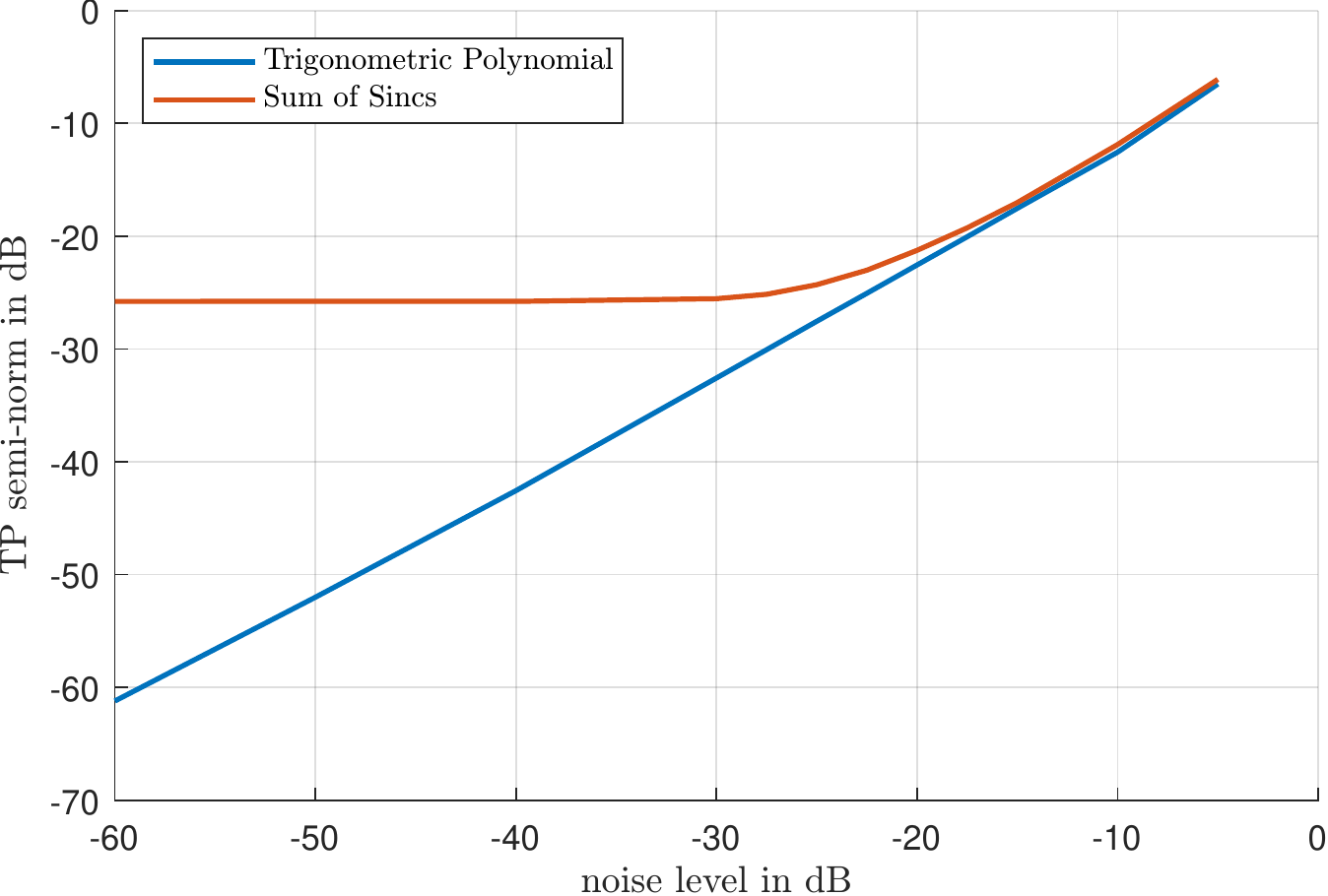}
    \subcaption{Operator norm on $\mathcal P_{N}$.}
  \end{subfigure}
  \qquad
  \begin{subfigure}[c]{0.4\linewidth}
  \includegraphics[width=\linewidth]{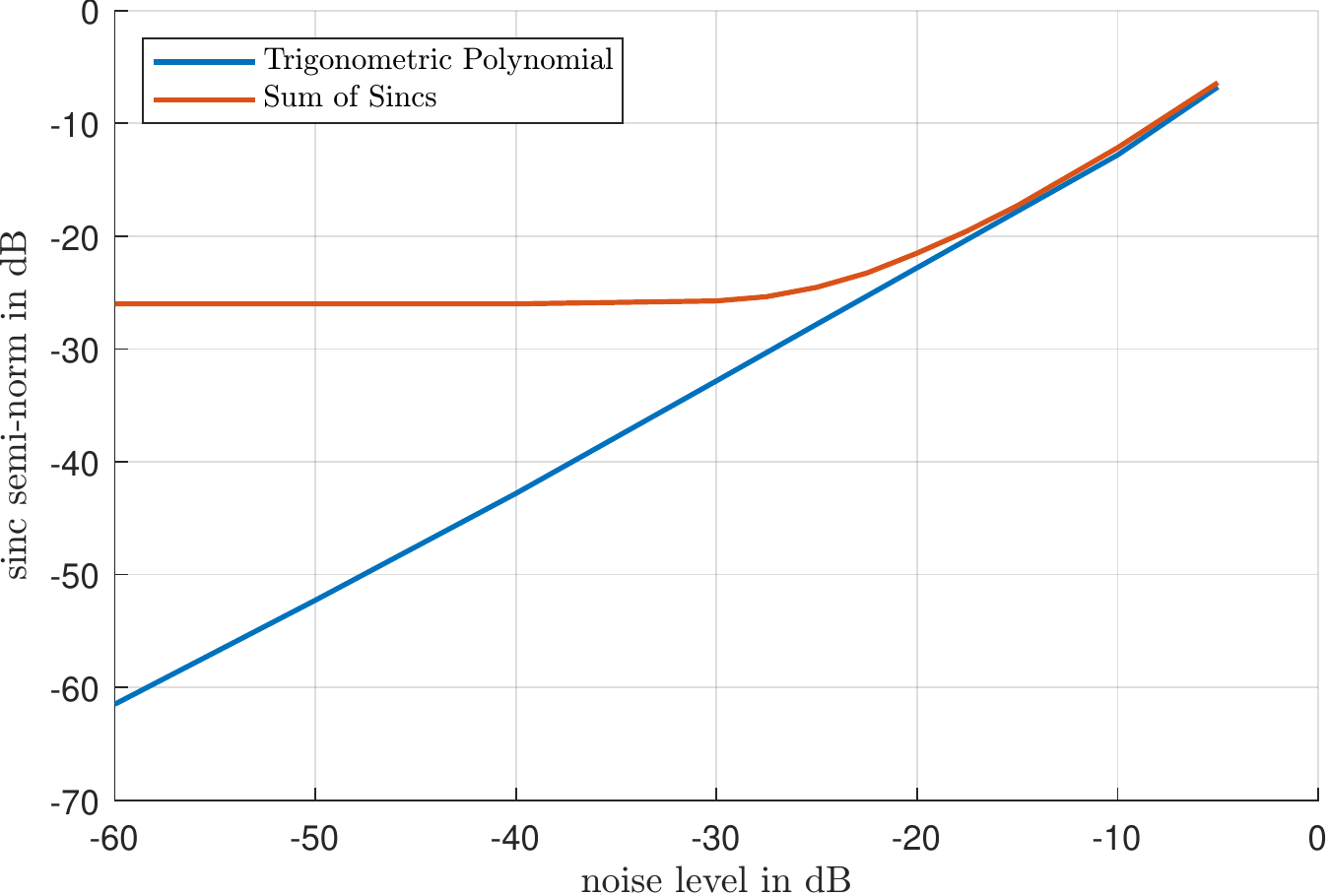}
    \subcaption{Operator norm on $\mathcal S_L$.}
  \end{subfigure}
  \caption{Influence of the identifier model on the reconstruction
    error depending on additive Gaussian noise.  For each noise level,
    50 channels with 10 features have been recovered from 101
    samples. The regularization parameter has been chosen proportional
    to the noise level.}
  \label{fig:model-miss}
\end{figure}

The replacement of the trigonometric polynomial by a sum of sinc
functions leads to a model error.  Considering a channel with 10
features and 101 samples as before, and studying the recovery error of
Algorithm~\ref{alg:mp_gd} measured in the operator norm, we see that
the model mismatch corresponds to a noise level of around
$-25~\mathrm{db}$.  Notice that the comparison with respect to
trigonometric polynomials is somehow subjective.  For this reason, we
also compute the relative reconstruction error based on the subspace
of sinc functions \eqref{eq:sinc-sum}.  Numerically, the difference
between both error terms is negligible.  The clearly visible
approximation for sinc functions does not occur for trigonometric
identifiers.

\vspace{13pt}
\centerline{ACKNOWLEDGEMENT}
\vspace{13pt}

\noindent 
We thank Götz Pfander and Dae Gwan Lee for inspiring discussions.
This work was supported by Deutsche Forschungsgemeinschaft (DFG) grant
JU 2795/3.

\bibliographystyle{abbrv}
\bibliography{references}

\appendix

\section{Proof of Theorem \ref{lem:bandlimit_periodic_stieltjes}} 
\label{app:identifier}
\noindent
To prove Theorem \ref{lem:bandlimit_periodic_stieltjes} we need the
following auxiliary lemmata.  We start with Poisson's summation
formula for bandlimited functions. Since we have not found it directly
in the literature we give the proof for convenience.

\begin{lemma}[Poisson Summation Formula for Bandlimited
  $L^1$-Functions]\label{thm:poisson_bandlimited}
Let $f \in L^1(\R) \cap C_0(\mathbb R)$ be bandlimited.
Then, for $\alpha > 0$, the $\alpha$-periodic function $F_\alpha$ given by
\begin{equation}
F_\alpha(x) := \sum_{k \in \Z} f(x - \alpha k), \qquad x \in \R,
\end{equation}
converges absolutely for all $x \in \R$, and we have
\begin{equation}
\label{eq:poisson_bandlimited:1}
F_\alpha(x) = \frac{1}{\alpha} \sum_{k \in \Z} \hat f \left(\frac{k}{\alpha}\right) \e^{2 \pi i \frac{k x}{\alpha}}.
\end{equation}
In particular, $F_\alpha$ is a trigonometric polynomial for all $\alpha > 0$.
\end{lemma}

\begin{proof}
  By assumption, we have $\supp \hat f \subseteq [-\sigma, \sigma]$
  for some $\sigma > 0$ so that we may identify $f$ as an element in
  $B^1_{2 \pi \sigma}$.  By Theorem \ref{thm:nikolskiis_inequality},
  we know that
\begin{equation}
\sum_{k \in \Z} | f(x - \alpha k)|
\leq \frac{1 + 2 \pi \sigma \alpha}{\alpha} \|f\|_1
\label{eq:poisson_bandlimited:2}
\end{equation}
for all $x \in \R$.  This shows that $F_\alpha$ is indeed well-defined
and bounded.  In particular,
$F_\alpha \in L^\infty(\R / \alpha \Z) \subset L^1(\R / \alpha \Z)$
and we can compute the Fourier coefficients
\begin{align}
\hat F_\alpha(k)
&= 	\frac{1}{\alpha} \int_0^\alpha 
		\sum_{\ell \in \Z}
			f(x - \alpha \ell)
		\e^{-2 \pi i x \frac{k}{\alpha}} 
	\intd x
= 	\frac{1}{\alpha} \sum_{\ell \in \Z}
		\int_0^\alpha 
			f(x - \alpha \ell)
		\e^{-2 \pi i x \frac{k}{\alpha}} 
\intd x
\label{eq:poisson_bandlimited:3}\\
&= 	\frac{1}{\alpha} \sum_{\ell \in \Z}
		\int_{\alpha \ell}^{\alpha (\ell + 1)} 
			f(x)
			\e^{-2 \pi i x \frac{k}{\alpha}}
		\intd x  
		= 	\frac{1}{\alpha} \int_\R
		f(x) \e^{-2 \pi i x \frac{k}{\alpha}}
	\intd x
  = \tfrac 1\alpha \hat f \left(\tfrac{k}{\alpha}\right).
\end{align}
Interchanging the series and integral in \eqref{eq:poisson_bandlimited:3} 
is allowed by the theorem of Fubini--Tonelli since $x \mapsto \sum_{\ell \in \Z} |f(x - \alpha \ell)|$ 
is uniformly bounded by \eqref{eq:poisson_bandlimited:2} and thus integrable on $[0, \alpha]$.

Since $\hat f$ has compact support, only finitely many Fourier
coefficients are non-zero, so the Fourier series
\begin{equation}
F_\alpha(x) = \frac{1}{\alpha} \sum_{k \in \Z} \hat f \left(\frac k\alpha \right) \e^{2 \pi i \frac{k x}{\alpha}}
\end{equation}
converges uniformly and is indeed a trigonometric polynomial.
\end{proof}

\begin{lemma}\label{lem:stieltjes_bandlimit_smoothness}
  Let $f = \hat \mu_f$, where $\mu_f \in \M(\R)$ fulfills
  $\supp \mu_f \subseteq [-\sigma, \sigma]$ for some $\sigma > 0$.
  Then $f$ is infinitely often differentiable and
  $f^{(n)} = \hat \mu_{f^{(n)}} \in \M(\R)$ for all $n \in \mathbb N$
  with
\begin{equation}\label{eq:stieltjes_bandlimit_smoothness:1}
\mu_{f^{(n)}}
= (-2 \pi i \cdot)^n \mu_f.
\end{equation}
In particular,  $\supp \mu_{f^{(n)}} \subseteq [-\sigma, \sigma]$.
\end{lemma}

\begin{proof}
Consider the difference quotients 
$g_h(x, \xi) \coloneqq \frac{1}{h}(\e^{-2\pi i (x + h) \xi} - \e^{-2 \pi i x
  \xi})$ for $x \in \R$, $\xi \in [-\sigma, \sigma]$ and $h \neq 0$.
Due to the mean value theorem,
they are uniformly bounded by
\begin{equation}
|g_h(x, \xi)|
= | \tfrac{1}{h}(\e^{-2\pi i (x + h) \xi} - \e^{-2 \pi i x \xi}) |
\leq \sup_{x^\ast \in [x, x + h]} |(-2 \pi i \xi) \e^{-2 \pi i x^\ast \xi} |
\leq 2 \pi \sigma.
\end{equation}
Since constant functions are integrable w.r.t.\ $\mu_f \in \M(\R)$, 
it follows from the dominated convergence theorem that 
\begin{equation}\label{eq:stieltjes_bandlimit_smoothness:2}
f'(x) 
= \lim_{h \to 0} 
	\int_{-\sigma}^{\sigma} 
		g_h(x, \xi) 
	\intd \mu_f(\xi)
=  	\int_{-\sigma}^{\sigma} 
		\lim_{h \to 0} g_h(x, \xi) 
	\intd \mu_f(\xi)
= \int_{-\sigma}^{\sigma} 
		(-2 \pi i \xi)
		\e^{-2\pi i x \xi} 
	\intd \mu_f(\xi).
\end{equation}
Repeating the above argument starting with $f'$, then $f^{(2)}$, and
so forth, we obtain the claim inductively for all $n \in \N$.
\end{proof}
                                                               
\begin{lemma}[{\cite[Thm~4.4, p.\ 25]{Katznelson04}}]
\label{thm:fourier_decay}
Let $f$ be an infinitely often differentiable, $T$-periodic function for $T > 0$.
Denote the Fourier coefficients of $f$ by
$$\hat f(k) = \tfrac1T \int_0^T f(x) \e^{-2\pi i \frac{k x}{T}} \intd x.$$
Then for all $j \in \N_0$ there exists $C_j > 0$ such that
\begin{equation}
|\hat f(k)| \leq C_j |k|^{-j} \qquad \text{for all } k \in \Z.
\end{equation}
\end{lemma}
\medskip

\noindent
\emph{Proof of Theorem} \ref{lem:bandlimit_periodic_stieltjes}.  By
Lemma \ref{lem:stieltjes_bandlimit_smoothness} we know that $f$ is
infinitely often differentiable and by Lemma \ref{thm:fourier_decay}
we have for all $j \in \N$ that $|\hat f(k)| \leq C_j |k|^{-j}$ for
some $C_j > 0$, so in particular
$(\hat f(k))_{k \in \Z} \in \ell^1(\Z)$.  Define the Borel measure
$\mu$ by
\begin{equation}
\mu \coloneqq \frac{1}{n} \sum_{k \in \Z} \hat f(k) \delta \Bigl(\cdot + \frac{k}{n} \Bigr).
\end{equation}
We have to show that $\mu \in \M(\R)$ and we will use that $(\M(\R), \|\cdot\|_\mathrm{TV})$ 
is the dual space of $(C_0(\R), \|\cdot\|_\infty)$.
Let $\varphi \in C_0(\R)$ be arbitrary, then
\begin{equation}
\Bigl| \int_\R \varphi(\xi) \intd \mu(\xi) \Bigr|
= \Bigl| \frac{1}{n} 
	\sum_{k \in \Z} 
		\hat f(k) 
		\varphi \Bigl(-\frac{k}{n}\Bigr)
\Bigr|
\leq \frac{1}{n} 
	\bigl\| (\hat f(k))_{k \in \Z} \bigr\|_{\ell^1} 
	\bigl\| \varphi \bigr\|_\infty.
\end{equation} 
This shows that $\mu$ indeed defines a continuous linear functional on $C_0(\R)$.
The Fourier transform of $\mu$ is
\begin{equation}
\hat \mu(x)
= \frac{1}{n} \sum_{k \in \Z} \hat f(k) \e^{2\pi i \frac{kx}{n}}
= f(x),
\qquad x\in \R.
\end{equation}
Since the Fourier transform is unique this implies $\mu = \mu_f$.
Finally, by assumption $\supp \mu = \supp \mu_f \subseteq [-\sigma, \sigma]$, so that 
$\hat f(k) = 0$ for all $k \in \Z$ satisfying $|k| > \sigma n$ and 
we obtain \eqref{eq:bandlimit_periodic_stieltjes:1}.
This concludes the proof.  
\hfill $\Box$

\section{Proof of Theorem \ref{lem:sampling_theorem_l1}}
\label{app:b}
\noindent
\begin{proof}[Proof of Theorem \ref{lem:sampling_theorem_l1}]
  The first part can be proved exactly following the lines of the
  classical sampling theorem of Shannon, Whittaker, Kotelnikov, see
  \cite[Thm.~2.29] {PPST2019} for instance.  It remains to show the
  convergence in $L^1(\mathbb R)$.  Applying
  Theorem~\ref{thm:nikolskiis_inequality} to $\phi$, we obtain
  \begin{align}
    \sum_{|k| > M} |f(ak)||\phi(x-ak)|
    &\le
      \sup_{|k|>M} \{|f(ak)|\} \,
      \sum_{k \in \mathbb Z} |\phi(x-ak)|
    \\
    &\le
      \sup_{|k|>M} \{|f(ak)|\} \,
      \tfrac{1 + \pi}a \, \| \phi \|_1.
  \end{align}
  Since the right-hand side vanishes for $M \to \infty$ independently
  of $x$ due to $f \in C_0 (\mathbb R)$, the pointwise convergent
  series $\sum_{k \in \Z} |f(ak)||\phi(x-ak)|$ also converges
  uniformly.  The partial sums are continuous functions such that the
  limit is continuous too and, in particular, measurable.
  Using Levi's monotone convergence theorem \cite[Thm.~2.4.1]{Cohn13},
  we have
  \begin{align}
    \int_\R \Bigr| \sum_{|k| > M} f(ak) \phi(x - ak) \Bigr| \intd x
    &\leq  \int_\R \sum_{|k| > M} |f(ak)| |\phi(x - ak)| \intd x\\
    &= \sum_{|k| > M} |f(ak)| \int_\R |\phi(x - ak)| \intd x\\
    &= \|\phi\|_1 \sum_{|k| > M} |f(ak)|.
  \end{align}
  Since Theorem~\ref{thm:nikolskiis_inequality} ensures
  $\bigl(f(a k) \bigr)_{k \in \Z} \in \ell^1(\Z)$, the last expression converges to
  zero as $M \rightarrow \infty$, which establishes the
  $L^1$-convergence.
\end{proof}

\end{document}